\documentclass[12pt]{article}
\usepackage{algorithm,algorithmic,amsmath,amssymb,epsf,epsfig,amsthm,times,ifthen,epic,psfrag,array,etoolbox}
\usepackage{enumerate,subcaption,multirow,makecell,enumitem,tikz,pgfplots}
\usepackage{printtime}
\usepackage{cite}     
\usepackage{fancyhdr}
\usepackage{setspace} 
\usepackage{lastpage} 
\usepackage{url}
\usepackage{hyperref} 
\usepackage{refcount} 

\newcommand{\CoZeNonLinearLemmaFourPointSix}{4.6}  

\def\submitteddate{August 5, 2016}
\def\reviseddate{}

\renewcommand{\baselinestretch}{1.0}
\DeclareFontFamily{U}{matha}{\hyphenchar\font45}
\DeclareFontShape{U}{matha}{m}{n}{
      <5> <6> <7> <8> <9> <10> gen * matha
      <10.95> matha10 <12> <14.4> <17.28> <20.74> <24.88> matha12
      }{}
\DeclareSymbolFont{matha}{U}{matha}{m}{n}
\DeclareMathSymbol{\notdivides}{3}{matha}{"1F}
\DeclareMathSymbol{\divides}{3}{matha}{"17}

\begin{document}

\newcommand{\creationtime}{\today \ \ @ \theampmtime}

\pagestyle{fancy}
\renewcommand{\headrulewidth}{0cm}
\chead{\footnotesize{Connelly-Zeger}}
\rhead{\footnotesize{\reviseddate}}
\lhead{}
\cfoot{Page \arabic{page} of \pageref{LastPage}} 

\renewcommand{\qedsymbol}{$\blacksquare$} 


\newtheorem{theorem}              {Theorem}     [section]
\newtheorem{lemma}      [theorem] {Lemma}
\newtheorem{corollary}  [theorem] {Corollary}
\newtheorem{proposition}[theorem] {Proposition}
\newtheorem{remark}     [theorem] {Remark}
\newtheorem{conjecture} [theorem] {Conjecture}
\newtheorem{example}    [theorem] {Example}

\theoremstyle{definition}         
\newtheorem{definition} [theorem] {Definition}
\newtheorem*{claim} {Claim}
\newtheorem*{notation}  {Notation}



\setlist[itemize]{itemsep=0.5pt, topsep=1pt}

\newcommand{\PartOne}{Part I}                           
\newcommand{\PartTwo}{Part II}                          

\newcommand{\Z}{\mathbf{Z}}
\newcommand{\N}{\mathbf{N}}

\newcommand{\alphabet}{\mathcal{A}}
\newcommand{\A}{\mathcal{A}}
\newcommand{\F}{\mathbb{F}}
\newcommand{\Capacity}{\mathcal{C}}

\newcommand{\PartitionRing}[2]{\mathrm{R}_{#1,#2}}

\newcommand{\pA}{\mathrm{A}}
\newcommand{\pB}{\mathrm{B}}
\newcommand{\pC}{\mathrm{C}}
\newcommand{\GF}[1]{\mathrm{GF}\!\left(#1\right)}
\newcommand{\Char}[1]{\mathsf{char}\!\left(#1\right)}
\newcommand{\Comment}[1]{ \left[\mbox{from  #1} \right]}
\newcommand{\Div}[2]{ #1  \bigm|  #2  } 
\newcommand{\NDiv}[2]{ #1   \notdivides   #2   }
\newcommand{\GCD}[2]{ \mathsf{gcd} \!\left( #1  ,  #2 \right)  } 
\newcommand{\PrimeFact}[1]{p_1^{k_1} \cdots p_{#1}^{k_{#1}}}

\newcommand{\pDivSymbol}{\bigm|}
\newcommand{\pDiv}[2]{\Div{#1}{#2}}
\newcommand{\pdiv}{\text{partition division}}
\newcommand{\partEquiv}{\equiv}

\newcommand{\SL}{scalar linearly}

\newcommand{\LSrdom} {\preceq}
\newcommand{\LSeq} {\equiv}

\newcommand{\mRings} [1] {\mathcal{R}(#1)}
\newcommand{\LinearSolvableNetworks} [1] {\mathcal{N}_{\text{lin}}(#1)}

\renewcommand{\aa} [2] {a_{#1,#2}}
\newcommand{\bb} [2] {b_{#1,#2}}

\newcommand{\osum}{\displaystyle\bigoplus}
\newcommand{\dsum}{\displaystyle\sum}

\newcommand{\DP}{\times}
\newcommand{\bigDP}{\displaystyle\prod}

\renewcommand{\emptyset}{\varnothing} 
\newcommand{\Network}{\mathcal{N}}
\newcommand{\TBA}{*** To Be Added ***}

\newcommand{\Fig}[2]{
 \medskip
  \epsfysize=#2 
  \epsffile{#1.eps}
 \medskip
}

\newcommand{\quasiorder} {\preccurlyeq}

\newcommand{\NChooseTwoNetwork}[1]{#1-Choose-Two Network}
\newcommand{\PolyRing}{\GF{p}[x]/\langle x^2\rangle}
\newcommand{\PolyRingMod}{(\PolyRing)/\langle x \rangle}
\newcommand{\GFpCrossGFp}{\GF{p}\times\GF{p}}

\let\bbordermatrix\bordermatrix
\patchcmd{\bbordermatrix}{8.75}{4.75}{}{}
\patchcmd{\bbordermatrix}{\left(}{\left[}{}{}
\patchcmd{\bbordermatrix}{\right)}{\right]}{}{}

\setcounter{page}{0}

\title{Linear Network Coding over Rings\\ \PartOne: Scalar Codes and Commutative Alphabets
\thanks{This work was supported by the 
National Science Foundation.\newline
\indent \textbf{J. Connelly and K. Zeger} are with the 
Department of Electrical and Computer Engineering, 
University of California, San Diego, 
La Jolla, CA 92093-0407 
(j2connelly@ucsd.edu and zeger@ucsd.edu).
}}

\author{Joseph Connelly and Kenneth Zeger\\}


\date{
\textit{
IEEE Transactions on Information Theory\\
Submitted: \submitteddate\\
}}

\maketitle
\begin{abstract}
  Fixed-size commutative rings are quasi-ordered such that all scalar linearly solvable networks 
  over any given ring are also scalar linearly solvable over any 
  higher-ordered ring.
  As consequences, if a network has a scalar linear solution over some finite commutative ring, then 
  (i) the network is also scalar linearly solvable over a maximal commutative ring of the same size, and
  (ii) the (unique) smallest size commutative ring 
  over which the network has a scalar linear solution is a field.
  We prove that a commutative ring is maximal with respect to the quasi-order if and only if
  some network is scalar linearly solvable over the ring
  but not over any other commutative ring of the same size.
  Furthermore, we show that maximal commutative rings are direct products of certain fields
  specified by the integer partitions of
  the prime factor multiplicities of the maximal ring's size.

  Finally, we prove that there is a unique maximal commutative ring of size $m$
  if and only if each prime factor of $m$ has multiplicity in $\{1,2,3,4,6\}$.
  In fact, whenever $p$ is prime and $k \in \{1,2,3,4,6\}$, 
  the unique such maximal ring of size $p^k$ is the field $\GF{p^k}$.
  However, for every field $\GF{p^k}$ with $k\not\in \{1,2,3,4,6\}$,
  there is always some network that is not 
  scalar linearly solvable over the field but is scalar linearly solvable over a commutative ring of the same size.
  These results imply that for scalar linear network coding over commutative rings, 
  fields can always be used when the alphabet size is flexible,
  but alternative rings may be needed when the alphabet size is fixed.

\end{abstract}

\thispagestyle{empty}

\clearpage


\clearpage
\section{Introduction} \label{sec:intro}
Linear coding over finite fields has been the cornerstone of a large portion
of network coding research during the last decade.  
Scalar linear codes over fields consist of network out-edges carrying field 
elements which are linear combinations of their input field elements.  
It has been known that scalar linear codes over finite fields are sufficient for multicast networks \cite{Li-Linear}.
This means that whenever a multicast network is solvable, it must be scalar linearly
solvable over some finite field.  
In contrast, the more general class of vector linear codes over fields have out-edges 
carrying linear combinations of input vectors of field elements, 
where the linear combination coefficients are matrices of field elements.
Vector linear codes over finite fields
(or even more generally, vector linear codes over rings or linear codes over modules) 
are known to not always
be sufficient for non-multicast networks \cite{DFZ-Insufficiency}.
This means that solvable non-multicast networks may sometimes require non-linear codes 
to implement a solution,
no matter what field or vector dimension is chosen.
Even though linear network codes may be suboptimal for some networks,
they have been attractive to study for two primary reasons: 
\begin{itemize}
  \item[(1)] They can be less complex to implement in practice due to reduced storage and/or
reduced computation compared to non-linear codes.
  \item[(2)] They may be mathematically tractable to analyze.
\end{itemize}

One of the most general forms of linear network coding uses codes over modules.
Specifically,
a module consists of an Abelian group $(G,\oplus)$, a ring $R$, 
and a scalar multiplication $\cdot :R \times G \to G$ 
that together 
satisfy certain properties.
A linear network code over such a module consists of edge functions of the form
$$(M_1 \cdot x_1) \oplus \cdots \oplus (M_m \cdot x_m)$$
where the variables $x_1, \dots, x_m$ are elements of $G$ and represent input symbols to a network node,
and the multiplier coefficients $M_1, \dots, M_m$ are constant elements of $R$.%
\footnote{Throughout this paper it will be assumed that rings always have multiplicative identities,
as any reasonable linear network code over rings would require.}
As an example,
vector linear network coding occurs when
$R$ is the ring of $n \times n$ matrices over a finite field,
$G$ is the set of $n$-dimensional vectors over the same field,
and $\cdot$ is matrix-vector multiplication over the field.
As another example, 
if $G$ is the additive group of the finite ring $R$
and $\cdot$ is multiplication in $R$,
then we get scalar linear coding over the ring alphabet $R$.

In this paper (i.e. \PartOne), we focus on the further special case where $R$ is a commutative ring,
and we make comparisons to the even more specialized (and more studied) case where $R$ is a field.
In a companion paper \cite{Connelly-RingNetworks-Part2} (i.e. \PartTwo),
we study vector linear codes and non-commutative rings
and specifically contrast the results with the results on scalar codes and commutative rings given
in this present paper.

Since the founding of network coding in 2000, 
network codes whose edge functions are 
linear over fixed finite field alphabets have been studied extensively
(e.g. \cite{Ebrahimi-Algorithms, Ho-Random, Jaggi-Algorithms, Karimian-Funnel, Koetter-Algebraic, 
Li-CommutativeAlgebra, Li-Linear, Li-LNC, Riis-LinearVsNonlinear, Sun-FieldSize, Sun-VL, Tavory-Bounds}).
In contrast, very little is presently known about linear network coding
over more general ring and module alphabets.

Since a field is a commutative ring  
that has inverses for all its non-zero elements, 
a linear network code over a ring may be implemented analogously to a linear code over a field, 
by performing multiplications and additions over the ring for each nontrivial edge function.%
\footnote{The most efficient implementation of ring arithmetic generally depends on the specific
algebraic properties of the ring being used.}
It is natural, then, to ask whether it is better 
in some sense to use linear coding over a finite field alphabet 
or over some ring alphabet of the same size that is not a field. 
Additionally, a finite field alphabet must have prime-power size,
so linear codes over rings 
may be of value if non-power-of-prime alphabet sizes are required.

A network is linearly solvable if all of its receivers can
linearly recover all of the messages they demand by using each network edge for at most one symbol transmission, 
where each such transmission is computed as a linear function of the inputs of the edge's parent node.
Many networks evolve over time as nodes
are added or deleted and as edge connections are formed or broken.
Thus, it might be advantageous to choose a coding alphabet that
makes as many networks as possible scalar linearly solvable over the chosen ring.
If, for example, 
every network that is scalar linearly solvable over a particular ring is also scalar linearly solvable
over a second ring, then, generally speaking, the second ring would be at least as good as the first ring.
This notion of one ring being better than another ring is the core concept behind
our study in this paper. 
We seek out the best such rings, namely the ones that are
maximal with respect to this induced ordering of rings.

Many interesting questions regarding linear codes over rings exist:
What is the best ring alphabet of a given size to use for linear network coding?
Are finite fields always the best choice?
Can a network be scalar linearly solvable over a ring,
even though it is not scalar linearly solvable over the field of the same size?
Is the set of networks that are scalar linearly solvable over some field
a proper subset of the set of networks that are scalar linearly solvable over some ring?
For alphabets whose sizes are not powers of primes, over which rings (if any) are particular networks
scalar linearly solvable?
We address these and some other questions in this paper.

Two of our main results are:
\begin{itemize}
\item[(1)] If $p$ is prime and $k\not\in\{1,2,3,4,6\}$,
then there always exists some network that is not scalar linearly solvable over the finite field $\GF{p^k}$
yet is scalar linearly solvable over a different commutative ring of the same size.
When $k\in\{1,2,3,4,6\}$, 
no such network exists.

\item[(2)] If a network has a scalar linear solution over a commutative ring 
that is not a field, 
then it also has a scalar linear solution over a field of strictly smaller size.
\end{itemize}

\subsection{Network model}\label{ssec:model}
A \textit{network} will refer to a finite, directed, acyclic multigraph,
some of whose nodes are \textit{sources} or \textit{receivers}.
Source nodes generate \textit{messages}, 
each of which is an arbitrary element 
of a fixed, finite set of size at least $2$,
called an \textit{alphabet}.
The elements of an alphabet are called \textit{symbols}.
The \textit{inputs} to a node are the messages, if any, originating at the node
and the symbols on the incoming edges of the node.
Each outgoing edge of a network node
has associated with it an \textit{edge function}
that maps the node's inputs
to the symbol carried by the edge, called the \textit{edge symbol}.
Each receiver node has \textit{decoding functions}
that map the receiver's inputs
to an alphabet symbol in an attempt to 
recover the receiver's \textit{demands},
which are the messages the receiver wishes to obtain.
A network is \textit{multicast} if there is a single source node
and each receiver demands every message.

In particular,
we will consider codes over alphabets
that have addition and multiplication operations,
namely finite rings.
If $\A$ is a ring alphabet,
then an edge function 
$$f:  \underbrace{\A \times \dots \times \A}_{m\ \text{inputs}}  \longrightarrow \A$$
is \textit{linear over $\A$} if it can be written in the form
\begin{align}
  f(x_1,\dots,x_m) & = M_1 x_1 + \dots + M_m x_m \label{eq:a}
\end{align}
where 
$M_1, \dots, M_m$ are constant values in $\A$.
A decoding function is linear
if it has a form analogous to \eqref{eq:a}.

%
A \textit{scalar code over an alphabet $\A$}
is an assignment of edge functions to all of the edges in a network and
an assignment of decoding functions to all of the receiver nodes in the network.
A code is \textit{scalar linear over $\A$} if each edge function and each decoding function
is linear over $\A$.
A \textit{solution over $\A$} is a code over $\A$ such that
each receiver's decoding functions recover each of its demands
from its inputs.
We say a network is \textit{solvable over $\A$} 
(respectively, \textit{scalar linearly solvable over $\A$})
if there exists a solution over $\A$ 
(respectively, scalar linear solution over $\A$),
and we say a network is \textit{solvable}
(respectively, \textit{scalar linearly solvable})
if it is solvable 
(respectively, scalar linearly solvable) 
over some alphabet.

In contrast, \textit{vector linear network codes}
have $k$-dimensional message vectors,
$k$-dimensional edge symbols,
and edge functions that are linear combinations of input vectors,
using matrices as coefficients.
Scalar linear codes are a special case of vector linear codes where $k = 1$.

\subsection{Related work} \label{ssec:previous_work}

Ahlswede, Cai, Li, and Yeung \cite{Ahlswede-NIF}
introduced network coding in 2000
and showed that it is possible to increase the information throughput of a network
by allowing nodes to transmit functions of their inputs,
as opposed to simply relaying their inputs.
Li, Yeung, and Cai \cite{Li-Linear} showed that
every solvable multicast network
is scalar linearly solvable over every sufficiently large finite field,
although it was shown in \cite{DFZ-Insufficiency} that non-multicast networks
may not have this property.
More generally, it was recently shown in \cite{Connelly-NonLinear} that for each composite number $m$,
there exists a network that is not linearly solvable over any module alphabet
yet is non-linearly solvable over an alphabet of size $m$.

Networks were demonstrated by Riis \cite{Riis-LinearVsNonlinear},
Rasala Lehman and Lehman \cite{Lehman-Complexity},
and in \cite{DFZ-LinearitySolvability} 
that are solvable non-linearly
but not scalar linearly over the same alphabet size.
Effros, El Rouayheb, and Langberg \cite{Effros-Index}
showed that network coding and index coding are equivalent
in a general setting, including with linear and non-linear codes.
It is not currently known whether there exists an algorithm
that determines if a network is solvable;
however, determining whether a network is scalar linearly solvable over a particular field
has been studied extensively.

Koetter and M\'{e}dard \cite{Koetter-Algebraic} showed that
for every network,
there exists a finite collection of polynomials,
such that for
every finite field $\F$,
the network is scalar linearly solvable 
over $\F$ if and only if
the polynomials have a common root in $\F$.
Conversely, it was shown in \cite{DFZ-Polynomials} 
that for every finite collection of polynomials,
there exists a network, such that
for every finite field $\F$,
the polynomials have a common root in $\F$
if and only if
the network is scalar linearly solvable over $\F$.
This connection between scalar linear solvability and polynomials
stems from the connection between scalar linearly solvable networks and matroid theory.
It was also shown in \cite{DFZ-Matroids} 
that every scalar linearly solvable network
is naturally associated with a representable matroid.

The study of linear network codes over fields
has led to efficient methods of constructing scalar linear solutions for networks
that also minimize alphabet size.
Ho et. al \cite{Ho-Random} described a random scalar linear coding technique
where the probability that a code is a solution grows with the field size.
Jaggi et. al \cite{Jaggi-Algorithms} presented polynomial-time algorithms
for designing scalar linear codes for multicast networks.
Karimian, Borujeny, and Ardakani \cite{Karimian-Funnel} 
showed there exists a class of non-multicast networks
for which random scalar linear coding algorithms fail with high probability
and presented a new approach to random scalar linear network coding for such networks.
Rasala Lehman and Lehman \cite{Lehman-Complexity}
and Tavory, Feder, and Ron \cite{Tavory-Bounds}
independently showed that some solvable multicast networks
asymptotically require finite field alphabets to be at least
as large as twice the square root of the number of receiver nodes
in order to achieve scalar linear solutions.
Sun, Yin, Zi, and Long \cite{Sun-FieldSize} demonstrated a class of multicast networks
that are scalar linearly solvable over certain fields
but not every larger field.

M\'{e}dard, Effros, Ho, and Karger \cite{Medard-NonMulticast}
showed that there can exist a network that is vector linearly solvable
but not scalar linearly solvable.
Sun et. al \cite{Sun-VL} demonstrated that,
while vector linear codes can outperform scalar linear codes in terms yielding solutions for general networks,
there can exist multicast networks that are not $k$-dimensional vector linearly solvable over $\GF{2}$
yet have scalar linear solutions over some field alphabet whose size is less than $2^k$.
Etzion and Wachter-Zeh \cite{Etzion-Wachter-Zeh-2016} bounded the reduction in field size
needed for a vector linear solution to a multicast network as compared to a scalar linear solution.
Ebrahimi and Fragouli \cite{Ebrahimi-Algorithms} 
presented algorithms for constructing vector linear codes
that achieve solutions not possible with scalar linear codes.

Outside of the context of the insufficiency of linear codes,
there has been little study of linear network codes over more general ring and module alphabets.
In this paper and its companion, we consider such linear codes
and compare them to the well-studied case of linear codes over fields.

\subsection{Our contributions}\label{ssec:contributions}

Some of the key results of this paper are highlighted below.
In this paper (i.e. \PartOne), we restrict attention
to network coding alphabets that are finite rings with at least two elements
and specifically focus on scalar linear codes over commutative rings with identity.
Our main results show that for networks that use scalar linear codes over commutative rings,
finite fields can always be used if the alphabet size is flexible,
but if the alphabet size is fixed, then finite fields may not always be the best choice for every network.

Section~\ref{sec:comparing}
introduces a ``dominance'' relation on finite rings, 
such that all networks that are scalar linearly solvable over a given ring
are also scalar linearly solvable over any ring that dominates the given ring.
We show that this relation is a quasi-order on the set of commutative rings of a given size.
We prove
(in Theorem~\ref{thm:R_dom_by_smaller_field})
that if a network has a scalar linear solution over some commutative ring,
then the unique smallest sized commutative ring over which the network has a scalar linear solution is a field.
Thus, for a given network, if the minimum alphabet size is desired
for scalar linear network coding,
it suffices to use finite fields.
This result also shows that networks that are scalar linearly solvable over some commutative ring 
are also scalar linearly solvable over some field
although not necessarily of the same size.
We also demonstrate
(in Theorem~\ref{thm:polyring}
and Corollary~\ref{cor:GF4xGF2=GF2xGF2xGF2})
non-isomorphic commutative rings of the same size that are equivalent with respect to dominance,
and we show (in Theorem~\ref{thm:Size-4-rings-linear-solvability}) that dominance is a
total quasi-order of the commutative rings of size $p^2$.

Section~\ref{sec:ChooseTwo} analyzes the scalar linear solvability of a class of multicast networks.
We show
(in Theorem~\ref{thm:FieldBeatsRing})
that for every finite field,
there exists a multicast network that is scalar linearly solvable over the field but 
is not scalar linearly solvable over any other commutative ring of the same size.
We also show (in Corollary~\ref{cor:fourchoosetwo_rings}) that
there exists a solvable multicast network 
that is not scalar linearly solvable over any ring whose size is $2n$, 
where $n$ is odd,
which contrasts with the fact that every solvable multicast network is scalar linearly solvable 
over every sufficiently large field.

Section~\ref{sec:ppo} compares various commutative rings with respect to dominance.
We demonstrate
(in Theorem~\ref{thm:onlyRing32})
that some network is scalar linearly solvable over a commutative ring of size $32$ but is not
scalar linearly solvable over any other commutative ring of size $32$, including the field $\GF{32}$,
and we later prove
(in Corollary~\ref{cor:32_min_size}) 
that $32$ is the size of the smallest such commutative ring alphabet where this phenomenon can occur.
We prove
(in Theorem~\ref{thm:field_div_ring})
that whenever a network is scalar linearly solvable over a commutative ring,
the network must also be scalar linearly solvable over a field whose size divides the ring size.
In fact, for each prime factor of the ring size, there is a corresponding such field 
whose characteristic equals the prime factor.
As a consequence
(in Corollary~\ref{cor:ring-size-6}),
whenever a network is scalar linearly solvable over a ring whose size is a product of distinct
primes (i.e. ``square free''), 
the network must also be scalar linearly solvable over each finite field whose size is a prime factor
of the ring size.
However, we demonstrate
(in Corollary~\ref{cor:rings-size-12})
that when a network is scalar linearly solvable over some commutative ring, 
the particular ring may need to be examined in order to
determine which fields the network is scalar linearly solvable over.

Section~\ref{sec:partition_div}
introduces partition rings,
which are direct products of finite fields
that are specified by integer partitions of the prime factor multiplicities of the ring size.
We define a relation called ``partition division''
and show that it induces a quasi-order on the set of partitions of a given integer.
We show that the maximal partitions under this quasi-order
are precisely the partitions that do not divide any other partition of the same integer.
We also provide a partial characterization of the maximal partitions.
The results of this section are used in various proofs in Section~\ref{sec:part_rings}.

Section~\ref{sec:part_rings} connects the relations of
ring dominance and partition division.
We prove 
(in Theorem~\ref{thm:maximal_rings})
that the maximal commutative rings under dominance
are precisely 
partitions rings where each partition is maximal under partition division.
We prove 
(in Theorem~\ref{thm:best_ring_for_network})
that a finite commutative ring is maximal if and only if
there exists a network that is scalar linearly solvable over the ring
but is not scalar linearly solvable over any other commutative ring of the same size.
Finally, we prove 
(in Theorem~\ref{thm:p5_7})
that if $p$ is prime, 
then the field $\GF{p^{k}}$ is the unique maximal commutative ring of size $p^{k}$
whenever $k \in \{1,2,3,4,6\}$,
but if $k = 5$ or $k \ge 7$,
then there exist multiple maximal commutative rings of size $p^{k}$.
This result is also generalized to commutative rings of non-power-of-prime sizes
in Theorem~\ref{thm:p5_7}.

Thus, since there can exist more than one maximal ring of a given size, 
there are instances where scalar linear solutions cannot be obtained using finite
field alphabets of a given size
but can be achieved using other commutative rings of the same size.

\PartTwo{} \cite{Connelly-RingNetworks-Part2}
studies similar network coding questions with emphasis
on non-commutative rings and vector linear codes.

\section{Comparison of rings for scalar linear network coding}\label{sec:comparing}

A \textit{quasi-order}%
\footnote{Also known as a \textit{pre-order} (e.g. \cite[Chapter 1]{Roitman-Set}).}
$\quasiorder$ on a set $A$ is a subset of $A\times A$ that is reflexive and transitive.
We write $x \quasiorder y$ to indicate that the pair $(x,y)$ is in the relation.
Each quasi-order induces an equivalence relation on $A$ defined by 
$x \equiv y$ if and only if $x \quasiorder y$ and $y \quasiorder x$.
We denote the equivalence class of $x$ by $[x]$.
Any quasi-order naturally extends to a partial order on the equivalence classes
by defining $[x] \quasiorder [y]$ if and only if $x \quasiorder y$.
An element $x\in A$ is said to be \textit{maximal} with respect to the quasi-order
if for all $y \in A$, we have $y \quasiorder x$ whenever $x \quasiorder y$.
The same definition of maximal applies with respect to the induced partial order on equivalence classes.

For each integer $m \ge 2$,
let $\mRings{m}$ denote the set of commutative rings of size $m$, up to isomorphism,
and let $\cong$ denote ring isomorphism.
For each finite ring $R$,
let $\LinearSolvableNetworks{R}$ be the set of all networks that are scalar linearly solvable over $R$.

For any two finite rings $R$ and $S$,
we say $S$ \textit{is dominated by} $R$ 
(denoted $S \LSrdom R$) if
every network that is scalar linearly solvable over $S$
is also scalar linearly solvable over $R$.
Equivalently, $S \LSrdom R$ if and only if $\LinearSolvableNetworks{S} \subseteq \LinearSolvableNetworks{R}$.
For each $m \ge 2$, it can be verified that 
the relation $\LSrdom$ is a quasi-order on the set $\mRings{m}$.
The induced equivalence relation on rings has the property that
$R \LSeq S$ if and only if 
$\LinearSolvableNetworks{R} = \LinearSolvableNetworks{S}$.
It turns out 
that the exact same set of networks can sometimes be scalar linearly solvable over non-isomorphic rings of the same size
(as illustrated later, in Theorem~\ref{thm:polyring} and Corollary~\ref{cor:GF4xGF2=GF2xGF2xGF2}),
which means that the quasi-order $\LSrdom$ is not anti-symmetric on $\mRings{m}$.
Throughout this paper, whenever we refer to a finite commutative ring as being maximal, we mean the ring is 
maximal with respect to the relation $\LSrdom$ on the set of commutative rings of the same size.
However, whenever we refer to a maximal ideal, we will always mean maximal with respect to set inclusion.

Intuitively, if a ring $R$ dominates a ring $S$ of the same size,
it may be viewed as advantageous%
\footnote{There may be other advantages to using one ring over another,
such as lower computational complexity arithmetic, ease of implementation, etc.} 
to use $R$ instead of $S$
in a network coding implementation,
since any network that is scalar linearly solvable over $S$ is also scalar linearly solvable over $R$,
and possibly even more networks are scalar linearly solvable over $R$, 
if $\LinearSolvableNetworks{S} \subset \LinearSolvableNetworks{R}$. 
A maximal commutative ring $R$ has the desirable property 
that, for any commutative ring $S$ of the same size, 
the set of networks that are scalar linearly solvable over $R$ 
cannot be a proper subset of the
set of networks that are scalar linearly solvable over $S$.
Thus, in this sense, maximal rings may be considered the ``best'' commutative rings to use for network coding,
and non-maximal rings are always ``worse'' than some maximal ring of the same size.
%

\subsection{Fundamental ring comparisons }\label{ssec:Basic}

\begin{lemma}
  Let $R$ and $S$ be finite rings.
  If $h: R \to S$ is a surjective homomorphism,
  then $R$ is dominated by $S$.
  \label{lem:RingHomomorphism}
\end{lemma}
\begin{proof}
  Let $\Network$ be a network that has a scalar linear solution over $R$.
  Every edge function in the solution is of the form
  \begin{align}
  y = M_1 x_1 + \cdots + M_m x_m
  \label{eq:100}
  \end{align}
  where the $x_i$'s are the parent node's inputs and the $M_i$'s are constants from $R$.
  Since $h$ is surjective, 
  for each symbol $x' \in S$
  there exists a symbol $x \in R$ such that $h(x) = x'$.

  Form a scalar linear code for $\Network$ over $S$ by replacing each coefficient $M_i$ in \eqref{eq:100} by $h(M_i)$.
  Suppose the inputs to the new edge function
  in the code over $S$ are $x_1', \dots, x_m' \in S$.
  Then, since $h$ is a homomorphism, the output of the edge function is
  \begin{align*}
    h(M_1) x_1' + \cdots + h(M_m) x_m'
    &= h(M_1) h(x_1) + \cdots + h(M_m) h(x_m) \\
    &= h(M_1 x_1 + \cdots + M_m x_m) \\
    &= h(y).
  \end{align*}
  Thus, whenever an edge function in the solution over $R$ outputs the symbol $y$,
  the corresponding edge function in the code over $S$ will output the symbol $h(y)$.
  Likewise, whenever $x$ is an input to an edge function in the solution over $R$,
  the corresponding input of the corresponding edge function in the code over $S$ 
  will be the symbol $x'$. 
  The same argument holds for the decoding functions in the code over $S$, so each
  receiver will correctly obtain its corresponding demands in the code over $S$,
  since $h(1) = 1$ and $h(0) = 0$.
  Thus, the code over $S$ is a scalar linear solution, and hence $R \LSrdom S$.
\end{proof}

In general, if a network is solvable (not necessarily linearly) over an alphabet $\alphabet$,
then it is also solvable over every alphabet of size $|\alphabet|^k$, for any $k \ge 2$,
by using a Cartesian product code.
The same fact is also true if we restrict to scalar linear codes.
In this sense, networks solvable over one alphabet are also solvable over certain larger alphabets.
In particular, if a network is scalar linearly solvable over the ring $\Z_n$,
then it is also scalar linearly solvable over the direct product of rings
$\Z_n^k = \underbrace{\Z_n \times \dots \times \Z_n}_{k\text{ times}}$.

Since $\Z_{n^k}$ is not isomorphic to the product ring $\Z_n^k$,
it does not immediately follow that a network scalar linearly solvable over $\Z_n$ must
also be scalar linearly solvable over $\Z_{n^k}$,
and, in fact, the contrary is demonstrated below in Corollary~\ref{cor:ints_mod_m}.

\begin{corollary}
  Let $m,n \ge 2$.
  The ring $\Z_{n}$ is dominated by the ring $\Z_{m}$
  if and only if
  $\Div{m}{n}$.
  \label{cor:ints_mod_m}
\end{corollary}
\begin{proof}
  Let $h: \Z_{n} \to \Z_{m}$ be defined such that
  $h(a)$ is the unique integer in $\{0,1,\dots,m-1\}$
  satisfying $h(a) = a$ mod $m$.
  If $\Div{m}{n}$, then
  $h$ is a surjective homomorphism,
  so by Lemma~\ref{lem:RingHomomorphism}
  we have $\Z_{n} \LSrdom \Z_{m}$.

  Conversely, if $\NDiv{m}{n}$, then
  the network $\Network_2(m,1)$%
  \footnote{Network $\Network_2(m,1)$ 
  is scalar linearly solvable over a finite ring $R$ 
  if and only if $\Div{\Char{R}}{m}$ (see \cite[Lemma~\CoZeNonLinearLemmaFourPointSix]{Connelly-NonLinear}).
  }
  of \cite{Connelly-NonLinear}
  is scalar linearly solvable over $\Z_m$
  but not $\Z_n$,
  since $\Char{\Z_m} = \Div{m}{m}$ and $\Char{\Z_n} = \NDiv{n}{m}$.
\end{proof}

If $p$ is prime and $k \ge 2$,
then by Corollary~\ref{cor:ints_mod_m}, 
we have
$\LinearSolvableNetworks{\Z_{p^k}} \subset \LinearSolvableNetworks{\Z_{p}}$.
In this sense, the larger ring alphabet $\Z_{p^k}$
is strictly ``worse'' than the smaller field alphabet $\Z_{p}$.

\begin{lemma}{\cite[Theorem 7, p. 243]{Dummit-Algebra}}
  If $I$ is a two-sided ideal of ring $R$, then
  the mapping $h: R \to R/I$
  given by
  $h(x) = x + I$
  is a surjective homomorphism.
  \label{lem:surjective_homomorphism}
\end{lemma}

\begin{corollary}
  If $I$ is an ideal in a finite commutative ring $R$,
  then $R$ is dominated by $R/I$.
  \label{cor:R_ideal}
\end{corollary}
\begin{proof}
  The quotient ring $R/I$
  is finite and commutative.
  By Lemma~\ref{lem:surjective_homomorphism},
  there is a surjective homomorphism from $R$ to $R/I$,
  so 
  $R \LSrdom R/I$ by Lemma~\ref{lem:RingHomomorphism}.
\end{proof}

Theorem~\ref{thm:R_dom_by_smaller_field} next demonstrates 
that when attempting to find a minimum size commutative ring
over which a network is scalar linearly solvable,
it suffices to restrict attention to finite field alphabets.
In other words,
if $\Network\in \LinearSolvableNetworks{R}$ for some commutative ring $R$,
then there exists a field $\F$ such that 
$\Network\in \LinearSolvableNetworks{\F}$
and $\Network\not\in \LinearSolvableNetworks{S}$ 
whenever $S \in \mRings{n}-\{\F\}$ 
and $n \le |\F|$.

\begin{theorem}
  If a network is scalar linearly solvable over a commutative ring,
  then the unique smallest such ring is a field. 
  \label{thm:R_dom_by_smaller_field} 
\end{theorem}
\begin{proof}
  Let $\Network$ be a scalar linearly solvable network and
  let $R$ be a smallest commutative ring over which $\Network$ is scalar linearly solvable.
  Suppose $R$ is not a finite field,
  and let $I$ be a maximal ideal of $R$.
  Then $R/I$ is a field (e.g. \cite[p. 254, Proposition 12]{Dummit-Algebra}).
  %
  By Lemma~\ref{lem:surjective_homomorphism},
  there is a surjective homomorphism from $R$ to $R/I$,
  but $R/I$ is a field and $R$ is not,
  so the rings cannot be isomorphic.
  Therefore, $|R/I| < |R|$.
  By Corollary~\ref{cor:R_ideal},
  $R \LSrdom R/I$.
  Thus $\Network$ must also be scalar linearly solvable over $R/I$, 
  which contradicts the assumption that $R$ is a smallest commutative ring over which $\Network$ is scalar linearly solvable.
\end{proof}

\begin{lemma}
  A network is scalar linearly solvable over a finite direct product of finite rings
  if and only if the network is scalar linearly solvable over each of the rings in the product.
  \label{lem:direct_product}
\end{lemma}
\begin{proof}
  Let $R_1,\dots,R_m$ be finite rings.
  For each $j = 1,\dots,m$,
  the projection mapping $h_j: \displaystyle\bigDP_{i=1}^{m} R_i \to R_j$ 
  defined by $h_j(x_1, \dots, x_m) = x_j$ is a surjective homomorphism,
  so by Lemma~\ref{lem:RingHomomorphism},
  $$\bigDP_{i = 1}^{m} R_i \LSrdom  R_j\; \; \; \; \; (j = 1,\dots,m),$$
  and thus any network that is scalar linearly solvable over the product ring $\displaystyle\bigDP_{i=1}^{m} R_i$
  is also scalar linearly solvable over each ring $R_1,\dots,R_m$.

  Conversely, any network 
  that is scalar linearly solvable over each ring $R_1,\dots,R_m$,
  is clearly scalar linearly solvable over the product ring $\displaystyle\bigDP_{i=1}^{m} R_i$ by using 
  a Cartesian product code of the scalar linear solutions over each $R_1,\dots,R_m$.
\end{proof}

Lemma~\ref{lem:direct_prod2} demonstrates 
that if each ring in a collection of rings dominates 
at least one ring in a second collection of rings, 
then the direct product of the rings in the first collection dominates the direct product
of the rings in the second collection.

\begin{lemma}
  If each of the finite rings $S_1,\dots,S_n$
  is dominated by at least one of the finite rings $R_1,\dots,R_m$,
  then $S_1 \times \dots \times S_n$ is dominated
  by $R_1 \times \dots \times R_m$.
  \label{lem:direct_prod2}
\end{lemma}
\begin{proof}
  Let $\Network$ be a network that is scalar linearly solvable over $\displaystyle\bigDP_{j=1}^n S_j$.
  Let $i \in \{1,\dots,m\}$ and let $j$ be such that $S_j \LSrdom R_i$.
  By Lemma~\ref{lem:direct_product},
  $\Network$ is scalar linearly solvable over $S_j$,
  so 
  $\Network$ is scalar linearly solvable over $R_i$.
  Thus by Lemma~\ref{lem:direct_product},
  since $i$ was chosen arbitrarily,
  $\Network$ is also scalar linearly solvable over $\displaystyle\bigDP_{i=1}^m R_i$.
\end{proof}

\begin{lemma}
  If $S$ is a subring of a finite commutative ring $R$,
  then $S$ is dominated by $R$.
  \label{lem:subring}
\end{lemma}
\begin{proof}
  Let $\Network$ be a network that is scalar linearly solvable over $S$.
  Any scalar linear solution to $\Network$ over $S$ is also a scalar linear solution to $\Network$ over $R$.
  To see this, note that
  the value carried by every out-edge and every decoding function 
  in the network solution over $S$ (respectively, over $R$) 
  is a scalar linear combination of the network's messages over $S$ (respectively, over $R$).
  Thus, each of the messages over $R$ will be decoded linearly in the exact same way as they are over $S$,
  since $S$ and $R$ have the same additive and multiplicative identities.
\end{proof}

A special case of the previous lemma is when $R = \GF{p^k}$ and $S = \GF{p^m}$,
where $p$ is prime and $k,m$ are positive integers such that $\Div{m}{k}$
(e.g., \cite[Theorem 2.3.1]{Bini-FCR}).
%
We also remark that for finite rings $R_1$ and $R_2$,
the multiplicative identity of $R_1 \times R_2$ is in neither $R_1$ nor $R_2$,
so while $R_1$ and $R_2$ are isomorphic to subsets of $R_1 \times R_2$ that are closed under addition and multiplication,
neither is a subring of $R_1 \times R_2$.

The following theorem demonstrates
that for each prime $p$,
it is possible to have
two non-isomorphic commutative rings of size $p^2$,
such that the rings are equivalent under dominance
(i.e., the exact same set of networks are scalar linearly solvable over each of the two rings).

\begin{theorem}
For each prime $p$,
the rings $\PolyRing$ and $\GFpCrossGFp$ are each dominated by the other
but are not isomorphic.
\label{thm:polyring}
\end{theorem}
\begin{proof}
The rings are clearly not isomorphic since the only element of $\GFpCrossGFp$
whose square is zero is zero itself,
and in $\PolyRing$, the squares of both zero and $x$ are zero.
The field $\GF{p}$ is a subring of $\PolyRing$, 
so by Lemma~\ref{lem:subring}, 
$\GF{p} \LSrdom \PolyRing$.
On the other hand,
the mapping $h: \PolyRing \to \GF{p}$ given by $h(a + bx) = a$
is a surjective homomorphism,
so by Lemma~\ref{lem:RingHomomorphism},
$\PolyRing \LSrdom \GF{p}$.
Thus, $\GF{p} \LSeq \PolyRing$.
By Lemma~\ref{lem:direct_product},
$\GFpCrossGFp \LSeq \GF{p}$.
\end{proof}

\subsection{The $n$-Choose-Two Networks}\label{sec:ChooseTwo}
Figure~\ref{fig:ChooseTwo} shows a multicast network studied by Rasala Lehman and Lehman \cite{Lehman-Complexity},
which we call the \textit{\NChooseTwoNetwork{$n$}}.
This network will be used to illustrate various facts in what follows.
The network has two messages $x$ and $y$, 
intermediate edge symbols $\lambda_1, \dots, \lambda_n$,
and $\binom{n}{2}$ receivers.
Each receiver receives a unique pair of symbols $(\lambda_i, \lambda_j)$, where $i\ne j$,
and must decode both messages $x$ and $y$.
A variation of the \NChooseTwoNetwork{$4$}, called the \textit{Two-Six Network},
is given in Figure~\ref{fig:TwoSix}.
The Two-Six Network was used in \cite{DFZ-LinearitySolvability} to show that
a multicast network with a solution over a given alphabet size
might not have a solution over all larger alphabet sizes.

\psfrag{x,y}{$x,y$}
\psfrag{lam1}{$\lambda_1$}
\psfrag{lam2}{$\lambda_2$}
\psfrag{lamm}{$\lambda_{n-1}$}
\psfrag{lamn}{$\lambda_n$}
%
\begin{figure}[h]
  \begin{center}
    \leavevmode
    \hbox{\epsfxsize=.5\textwidth\epsffile{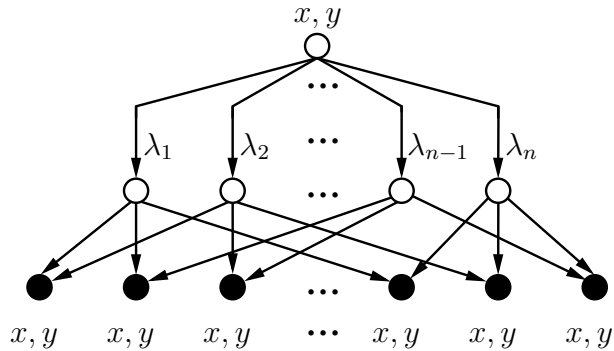}}
  \end{center}
  \caption{The \NChooseTwoNetwork{$n$} is parameterized by an integer $n \ge 2$.
  The network's name indicates the number of receivers.
  }
\label{fig:ChooseTwo}
\end{figure}

The following lemma characterizes the finite fields over which a 
scalar linear solution to the \NChooseTwoNetwork{$n$} exists
and gives an alphabet-size condition necessary for solvability.

\begin{lemma}{\cite[p. 144]{Lehman-Complexity}}  
  Let $\A$ be an alphabet and let $n \ge 3$. 
\begin{itemize}
    \item[(a)] If the \NChooseTwoNetwork{$n$} has a solution over $\A$,
      then $|\A| \geq n - 1$.
    \item[(b)]
     Let $\A$ be a field. 
     The \NChooseTwoNetwork{$n$} is scalar linearly solvable over $\A$
      if and only if
      $|\A| \ge n - 1$.
  \end{itemize}
\label{lem:kchoosetwo_fields}
\end{lemma}

The following theorem demonstrates that for each finite field,
there exists a multicast network that is scalar linearly solvable over the field
but is not scalar linearly solvable over any other commutative ring of the same size.

\begin{theorem}
  For each prime $p$ and positive integer $k$,
  the \NChooseTwoNetwork{($p^{k}+1)$}
  is scalar linearly solvable over the field $\GF{p^k}$
  but not over any other commutative ring of size $p^k$.
  \label{thm:FieldBeatsRing} 
\end{theorem}
\begin{proof}
Lemma~\ref{lem:kchoosetwo_fields} (b)
implies that the \NChooseTwoNetwork{$(p^k + 1)$}
is scalar linearly solvable over $\GF{p^k}$.
If the \NChooseTwoNetwork{$(p^k+1)$} network were scalar linearly solvable over a commutative ring $R$ 
of size $p^k$ that is not a field,
then by Theorem~\ref{thm:R_dom_by_smaller_field}
it would also be scalar linearly solvable over some field whose size is less than $p^k$,
which would contradict Lemma~\ref{lem:kchoosetwo_fields}.
\end{proof}

It is known \cite[Theorem 2, p. 250]{Fine-Ringsp2} that,
for each prime $p$,
the only four commutative rings of size $p^2$ are
$\GF{p^2}$,
$\GFpCrossGFp$,
$\Z_{p^2}$,
and 
$\PolyRing$.
The following theorem describes a chain of dominances between these rings
and shows that dominance is a total quasi-order of the commutative rings of size $p^2$.

\begin{theorem}
  For each prime $p$,
  the four commutative rings of size $p^2$ satisfy
  $$ 
  \LinearSolvableNetworks{\Z_{p^2}} 
  \subset \LinearSolvableNetworks{\PolyRing} 
  = \LinearSolvableNetworks{\GFpCrossGFp}
  \subset \LinearSolvableNetworks{\GF{p^2}}.
  $$
\label{thm:Size-4-rings-linear-solvability}
\end{theorem}
\begin{proof}
  The field $\GF{p}$ is a subring of the field $\GF{p^2}$,
  so by Lemma~\ref{lem:subring},
  $\GF{p} \LSrdom \GF{p^2}$.
  This, along with the fact the \NChooseTwoNetwork{$(p^2+1)$} is scalar linearly solvable over $\GF{p^2}$
  but not $\GF{p}$ (via Lemma~\ref{lem:kchoosetwo_fields}), implies
  $$\LinearSolvableNetworks{\GF{p}} 
  \subset \LinearSolvableNetworks{\GF{p^2}}.$$
  By Theorem~\ref{thm:polyring} and Corollary~\ref{cor:ints_mod_m},
  we also have
  $$\Z_{p^2} \LSrdom \GF{p} \LSeq \GFpCrossGFp \LSeq \PolyRing.$$
  Additionally, by Corollary~\ref{cor:ints_mod_m},
  there exists a network that is scalar linearly solvable over $\GF{p}$
  but not $\Z_{p^2}$,
  thus proving the claim.
\end{proof}

The following theorem gives a condition on the alphabet sizes over which a 
scalar linear solution to the \NChooseTwoNetwork{$n$}
exists for at least one commutative ring.
A converse is also given in terms of the network's scalar linear solvability over (not necessarily commutative) rings.

\begin{theorem}
    Let $m = \PrimeFact{t}$ denote the prime factorization of $m \ge 2$,
    and let $n \ge 3$.
  \begin{itemize}
    \item[(a)] If $p_i^{k_i} \ge n - 1$ for each $i = 1,\dots,t$, then
    the \NChooseTwoNetwork{$n$} is scalar linearly solvable over some commutative ring of size $m$.
    \item[(b)] If the \NChooseTwoNetwork{$n$} is scalar linearly solvable over some ring of size $m$, then
    $p_i^{k_i} \ge n - 1$ for each $i = 1,\dots,t$.
  \end{itemize}
\label{thm:kchoosetwo_rings}
\end{theorem}
\begin{proof}
Assume $p_i^{k_i} \ge n - 1$.
Then the \NChooseTwoNetwork{$n$} is scalar linearly solvable over
$\GF{p_i^{k_i}}$.
So by Lemma~\ref{lem:direct_product},
the \NChooseTwoNetwork{$n$} is scalar linearly solvable over the product ring
$\GF{p_1^{k_1}} \times \cdots \times \GF{p_t^{k_t}}$,
which has cardinality $m$.

Conversely, suppose $m = \PrimeFact{t}$
and the \NChooseTwoNetwork{$n$} is scalar linearly solvable over a ring $R$
of size $m$.
$R$ is isomorphic to a direct product of rings of size $p_1^{k_1},\dots,p_t^{k_t}$
(e.g. \cite[p. 2]{McDonald-FiniteRings}).
For each $i = 1,\dots,t$,
let $R_i$ be the ring of size $p_i^{k_i}$.
Then by Lemma~\ref{lem:direct_product},
the \NChooseTwoNetwork{$n$} is scalar linearly solvable over each of $R_1,\dots,R_t$.
Hence by Lemma~\ref{lem:kchoosetwo_fields} (a),
we must have $p_i^{k_i} \geq n -1$
for all $i$.
\end{proof}

\psfrag{x,y}{$x,y$}
\psfrag{lam1}{$\lambda_1$}
\psfrag{lam2}{$\lambda_2$}
\psfrag{lam3}{$\lambda_3$}
\psfrag{lam4}{$\lambda_4$}
%
\begin{figure}[h]
  \begin{center}
    \leavevmode
    \hbox{\epsfxsize=.5\textwidth\epsffile{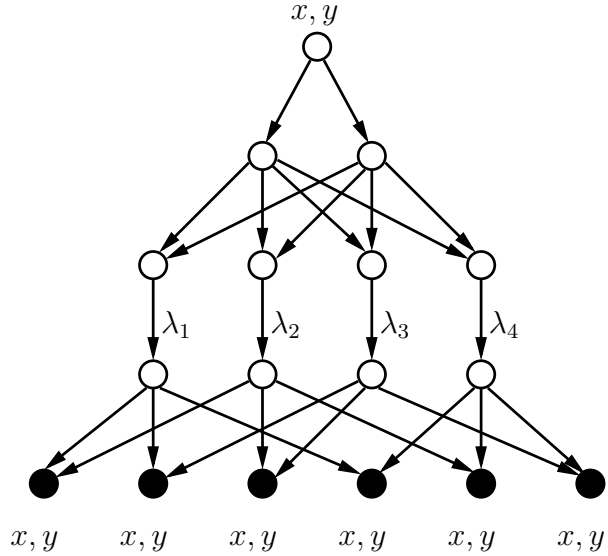}}
  \end{center}
  \caption{The Two-Six Network is a multicast network studied in \cite{DFZ-LinearitySolvability}.
  Each of the receivers gets a unique pair of edge symbols $(\lambda_i,\lambda_j)$, where $i\ne j$.
  The network's name indicates the alphabet sizes over which the network is not solvable.
  }
\label{fig:TwoSix}
\end{figure}

Corollary~\ref{cor:fourchoosetwo_rings} gives conditions on the solvability
and scalar linear solvability of the Two-Six Network.
We use the fact that the Two-Six Network 
is equivalent in terms of solvability to the \NChooseTwoNetwork{$4$}.

\begin{corollary}
    For each $m \ge 2$,
    the Two-Six Network is:
    \begin{itemize}
      \item[(a)] Solvable over an alphabet of size $m$
      if and only if $m \not\in \{2,6\}$.
      \item[(b)] Scalar linearly solvable over some ring of size $m$
      if and only if $m \ne 2$ mod $4$.
      \item[(c)] Scalar linearly solvable over all finite fields except $\GF{2}$.
    \end{itemize}
\label{cor:fourchoosetwo_rings}
\end{corollary}
\begin{proof}
Part (a) is \cite[Lemma V.3]{DFZ-LinearitySolvability}.
Parts (b) and (c) follow immediately from
Theorem~\ref{thm:kchoosetwo_rings}
and Lemma~\ref{lem:kchoosetwo_fields}, respectively,
when $n = 4$.
\end{proof}

The proof of Corollary~\ref{cor:fourchoosetwo_rings} (a)
(i.e. Lemma~V.3 in \cite{DFZ-LinearitySolvability})
made use of a theorem characterizing the orders for which orthogonal latin squares exist.
Euler originally conjectured over 230 years ago that orthogonal latin squares existed for all orders
not congruent to $2$ mod $4$.
It turned out that Euler was incorrect, and it was shown in 1960 that
orthogonal latin squares existed for all orders except $2$ and $6$.
Interestingly, the Two-Six Network 
was shown in Corollary~\ref{cor:fourchoosetwo_rings} 
to be solvable for all alphabet sizes except $2$ and $6$
and scalar linearly solvable over some ring of every size
that is not congruent to $2$ mod $4$.

Li, Yeung, and Cai \cite{Li-Linear} showed that
every solvable multicast network
is scalar linearly solvable over every sufficiently large finite field.
We observe that this property is not true for finite rings,
as the Two-Six Network is a solvable multicast network
and is not scalar linearly solvable over any ring of size $2n$,
where $n$ is odd.

\clearpage
\section{Finite field dominance}\label{sec:ppo}

A ring $R$ does not dominate the ring $S$ if and only if there exists a network
that is scalar linearly solvable over $S$ but not over $R$.
The following lemma demonstrates a class of networks
that will be used in later proofs to show a given ring is not dominated by another given ring.

\begin{lemma}{\cite[Section VI, Example (7)]{DFZ-Polynomials}}
  For any primes $q_1,\dots,q_s$ 
  and any positive integers $m_1,\dots,m_s$,
  there exists a network that 
  is scalar linearly solvable 
  over the fields
  $\GF{q_1^{nm_1}},$ $\dots$ $, \GF{q_s^{nm_s}}$ 
  for all $n \ge 1$, but not over any other fields.
  \label{lem:poly_network}
\end{lemma}

Note that the primes $q_1,\dots,q_s$ in Lemma~\ref{lem:poly_network} need not be distinct.
The following lemma will enable us to demonstrate certain networks that
are scalar linearly solvable over some ring of prime power size but not
over the field of the same size.
Lemma~\ref{lem:direct_prod3} will also be used in some of the proofs
in Section~\ref{sec:part_rings}.

\begin{lemma}
  Let $p_1,\dots,p_r$ and $q_1,\dots,q_s$ be primes,
  and let $k_1,\dots,k_r$ and $m_1,\dots,m_s$ be positive integers. 
  The ring $\GF{q_1^{m_1}} \times \dots \times \GF{q_s^{m_s}}$
  is dominated by 
  the ring $\GF{p_1^{k_1}} \times \dots \times \GF{p_r^{k_r}}$
  if and only if
  for each $i \in \{1,\dots,r\}$
  there exists $j \in \{1,\dots,s\}$
  such that $q_j = p_i$
  and $\Div{m_j}{k_i}$.
  \label{lem:direct_prod3}
\end{lemma}
\begin{proof}
  If for each $i$, there is a $j$ such that $q_j = p_i$ and $\Div{m_j}{k_i}$, then
  $\GF{q_j^{m_j}}$ is a subring 
  of $\GF{p_i^{k_i}}$ (e.g., \cite[Theorem 2.3.1]{Bini-FCR}),
  so by Lemma~\ref{lem:subring},
  $\GF{q_j^{m_j}} \LSrdom \GF{p_i^{k_i}}$,
  and therefore, by Lemma~\ref{lem:direct_prod2},
  $$\bigDP_{j=1}^s \GF{ q_j^{m_j}} \LSrdom \bigDP_{i=1}^r \GF{p_i^{k_i}} .$$

  To prove the converse, suppose to the contrary that 
  there exists $i \in \{1,\dots,r\}$
  such that for all $j \in \{1,\dots,s\}$,
  either $q_j \ne p_i$ or $\NDiv{m_j}{k_i}$.
  By Lemma~\ref{lem:poly_network},
  there exists a network $\Network$
  that is scalar linearly solvable precisely over those fields of size $q_j^{n m_j}$,
  where $j \in \{1,\dots,s\}$ and $n\ge 1$.
  Taking $n=1$ and applying Lemma~\ref{lem:direct_product},
  implies that $\Network$ is scalar linearly solvable over 
  $\bigDP_{j=1}^s \GF{ q_j^{m_j}}$.
  But $\Network$ can not be scalar linearly solvable over $\GF{p_i^{ k_i }}$,
  since for all $j \in \{1,\dots,s\}$, either $q_j \ne p_i$ or $\NDiv{m_j}{k_i}$,
  so by Lemma~\ref{lem:direct_product},
  $\Network$ is not scalar linearly solvable over 
  $\bigDP_{i=1}^r \GF{p_i^{k_i}}$.
  Thus, 
  $$\bigDP_{j=1}^s \GF{ q_j^{m_j}} \not\LSrdom \bigDP_{i=1}^r \GF{p_i^{k_i}}.$$
\end{proof}

As in Theorem~\ref{thm:polyring},
the following corollary demonstrates that two non-isomorphic commutative rings 
may be equivalent with respect to the dominance relation $\LSrdom$.
In this case, the rings are both direct products of fields.

\begin{corollary}
  For each $k \geq 3$ and prime $p$,
  the rings 
  $\GF{p^{k-1}} \DP \GF{p}$ and $\GF{p^{k-2}} \DP \GF{p} \DP \GF{p}$
  each dominate the other.
\label{cor:GF4xGF2=GF2xGF2xGF2}
\end{corollary}
\begin{proof}
The result follows from
Lemma~\ref{lem:direct_prod3}
by taking $r=2, s=3$, $p_1=p_2=q_1=q_2=q_3=p$, $k_1=k-1$, $m_1=k-2$, and $k_2=m_2=m_3 = 1$
to get $$\GF{p^{k-2}} \DP \GF{p} \DP \GF{p} \LSrdom \GF{p^{k-1}} \DP \GF{p}  $$
and by taking $r=3, s=2$, $p_1=p_2=p_3=q_1=q_2=p$, $k_1=k-2$, $m_1=k-1$, and $k_2=k_3=m_2 = 1$
to get $$\GF{p^{k-1}} \DP \GF{p} \LSrdom \GF{p^{k-2}} \DP \GF{p} \DP \GF{p}.$$
\end{proof}

Example~\ref{ex:size32ring} next demonstrates 
a network that 
is scalar linearly solvable over a ring of size $32$
but is not scalar linearly solvable over the field of size $32$.
It turns out that $32$ is the smallest prime power alphabet size
for which a network can have a scalar linear solution 
over a commutative ring but not over the field of the same size
(see Corollary~\ref{cor:32_min_size}).
%

\begin{example}
  Taking $r = 1, s = 2$, $p_1 = q_1 = q_2 = 2$ and $k_1 = 5, k_1 = 3, k_2 = 2$
  in Lemma~\ref{lem:direct_prod3}
  shows that $\GF{8} \times \GF{4}$ is not dominated by $\GF{32}$.
  In particular, there exists a network that is scalar linearly solvable over the ring $\GF{8} \times \GF{4}$
  but not over the field $\GF{32}$.
  \label{ex:size32ring}
\end{example}

\begin{lemma}{\cite[Theorem I.2]{DFZ-Polynomials}}
  For any collection of polynomials with integer coefficients,
  there exists a network such that for each finite field $\F$,
  the network is scalar linearly solvable over $\F$ if and only if there is
  an assignment of values from $\F$ to the variables in the polynomial collection
  such that each of the polynomials evaluates to zero.
\label{lem:DFZ-polynomials}
\end{lemma}

Sun, Yin, Li, and Long \cite{Sun-FieldSize} presented a class of multicast networks, 
called \textit{Swirl Networks}, 
parameterized by an integer $\omega \geq 3$
that affects the number of independent messages generated by the source
as well as the number of receivers and intermediate nodes.
Example~\ref{ex:size213ring} uses a particular case of the Swirl Network to
demonstrate that there exists a multicast network that is scalar linearly solvable over a ring of size $2^{13}$
but is not scalar linearly solvable over a field of the same size.

\begin{example}
  It was shown in \cite[p. 6185]{Sun-FieldSize} that
  the Swirl Network with $\omega = 2^{13}$ is
  scalar linearly solvable over $\GF{2^9}$ and $\GF{2^4}$
  but not over $\GF{2^{13}}$.
  Thus, this Swirl Network is scalar linearly solvable over the ring $\GF{2^9} \times \GF{2^4}$ of size $2^{13}$.
  A non-multicast network can be constructed with similar solvability properties
  to this Swirl Network, by using in Lemma~\ref{lem:DFZ-polynomials}
  a polynomial that is the product of irreducible
  polynomials of degrees $4$ and $9$ with binary coefficients,
  such as $x^4 + x + 1$ and $x^9 + x + 1$.
  %
  Likewise, taking $r = 1$, $s = 2$, $p_1 = q_1 = q_2 = 2$, $m_1=4$, $m_2=9$, and $k_1=13$ in
  Lemma~\ref{lem:direct_prod3}
  shows that $\GF{2^4} \times \GF{2^9}$ is not dominated by $\GF{2^{13}}$,
  so there exists a non-multicast network that is scalar linearly solvable over the ring $\GF{2^9} \times \GF{2^4}$
  but not over the field $\GF{2^{13}}$.
  \label{ex:size213ring}
\end{example}

By Lemma~\ref{lem:kchoosetwo_fields},
the \NChooseTwoNetwork{$n$} is scalar linearly solvable over finite field $\F$
if and only if $|\F| \ge n - 1$.
We note that the same property can be achieved in a different (non-multicast) network by 
applying Lemma~\ref{lem:DFZ-polynomials} to
Example (8) in \cite[p. 2315]{DFZ-Polynomials},
where for any $n\ge 2$, a polynomial is given that has a root in $\F$
if and only if $|\F| \ge n - 1$.

Theorem~\ref{thm:FieldBeatsRing} and Examples~\ref{ex:size32ring} and \ref{ex:size213ring}
also demonstrate that dominance is not necessarily a total quasi-order of the commutative rings of a given size,
as there can exist rings of the same size
such that neither dominates the other.
%

\subsection{Local rings}\label{ssec:local}

A finite commutative ring is said to be \textit{local} if it has a single maximal ideal
(see \cite[Definition 1.2.9]{Bini-FCR}).
%
Lemmas~\ref{lem:dp_local_rings} and \ref{lem:local_prime_power}
are standard results from commutative ring theory.

\begin{lemma}{\cite[Theorem 3.1.4]{Bini-FCR}}
  Every finite commutative ring is a direct product of local rings.
\label{lem:dp_local_rings}
\end{lemma}

\begin{lemma}{\cite[Theorem 6.1.2 II]{Bini-FCR}}
  If $R$ is a finite commutative local ring with maximal ideal $I$,
  then there exists a prime $p$
  and positive integers $k$ and $m$ such that
  \begin{itemize}
    \item[(i)] $|R| = p^k$
    \item[(ii)] $R/I$ is a field of size $p^m$ and $m$ divides $k$.
  \end{itemize}
  \label{lem:local_prime_power}
\end{lemma}
%

All finite fields are local rings, 
since their unique maximal ideal is the trivial ring $\{0\}$.
The ring $\Z_n$ is local if and only if $n$ is a prime power.
However, not every ring of prime power size is local
(e.g. $\GF{2} \times \GF{2}$ has distinct maximal ideals $\{(0,0),(1,0)\}$
and $\{(0,0),(0,1)\}$).

The following lemma connects the algebraic concept of local rings
to the dominance relation of network coding.

\begin{lemma}
  Every finite commutative local ring is dominated by the finite field of the same size.
  \label{lem:GF_dom_local}
\end{lemma}
\begin{proof}
  Let $R$ be a finite commutative local ring with maximal ideal $I$.
  By Lemma~\ref{lem:local_prime_power}, there exist a prime
  $p$ and positive integers $k$ and $m$
  such that $|R| = p^k$,
  \begin{align}
    R/I &\cong \GF{p^m}
      \label{eq:5} ,
  \end{align}
  and $\Div{m}{k}$.
  Thus,
  \begin{align*}
  R & \LSrdom \GF{p^m} && \Comment{\eqref{eq:5}, Corollary~\ref{cor:R_ideal}}\\
    & \LSrdom \GF{p^k} && \Comment{$\Div{m}{k}$, Lemma~\ref{lem:direct_prod3}}.
  \end{align*}
\end{proof}

Example~\ref{ex:size32ring} demonstrated that there exists a network 
that is scalar linearly solvable over the ring $\GF{8} \times \GF{4}$
but not over the field $\GF{32}$.
The following theorem strengthens the result in Example~\ref{ex:size32ring}
by additionally showing the network is not even
scalar linearly solvable over any other commutative ring of size $32$.
This contrasts with Theorem~\ref{thm:FieldBeatsRing},
which demonstrates a network that is scalar linearly solvable over $\GF{32}$
but not over any other commutative ring of size $32$.

\begin{theorem}
  There exists a network that is scalar linearly solvable over $\GF{8} \times \GF{4}$
  but not over any other commutative ring of size $32$.
  \label{thm:onlyRing32}
\end{theorem}
\begin{proof}
By Lemma~\ref{lem:poly_network},
there exists a network $\Network$
that is scalar linearly solvable precisely over all fields whose size is of the form
$2^{2n}$ or $2^{3n}$, where $n\ge 1$.
Hence $\Network$ is scalar linearly solvable over both $\GF{4}$ and $\GF{8}$
but neither $\GF{2}$ nor $\GF{32}$.
By using a product code, $\Network$ is also scalar linearly solvable over the ring $\GF{8} \times \GF{4}$ of size $32$.
We will now show that $\Network$ is not scalar linearly solvable over any other commutative ring of size $32$.

By Lemmas~\ref{lem:dp_local_rings} and \ref{lem:local_prime_power} (i),
every commutative ring $R$ of size $32$
satisfies exactly one of the following seven properties:
\begin{itemize}
  \itemsep0em
  \item[(a)] $R$ is a local ring of size $32$ 
  \item[(b)] $R$ is a direct product of local rings of size $16$ and $2$
  \item[(c)] $R$ is a direct product of local rings of size $8$ and $4$
  \item[(d)] $R$ is a direct product of local rings of size $8$, $2$, and $2$
  \item[(e)] $R$ is a direct product of local rings of size $4$, $4$, and $2$
  \item[(f)] $R$ is a direct product of local rings of size $4$, $2$, $2$, and $2$
  \item[(g)] $R$ is a direct product of local rings of size $2$, $2$, $2$, $2$, and $2$.
\end{itemize}
By Lemma~\ref{lem:GF_dom_local},
any network that is scalar linearly solvable over a commutative local ring of size $32$
is also scalar linearly solvable over $\GF{32}$.
This eliminates case (a).
Similarly, any network that is scalar linearly solvable over a local ring of size $2$
is also scalar linearly solvable over $\GF{2}$.
By Lemma~\ref{lem:direct_product},
any network that is scalar linearly solvable over a direct product ring
is also scalar linearly solvable over every ring in the direct product.
This eliminates cases (b),(d),(e),(f),(g).
Thus if $\Network$ is scalar linearly solvable over a commutative ring $R$ of size $32$,
$R$ must satisfy case (c).

Suppose $S$ is a commutative local ring of size $8$ with maximal ideal $I$.
Then Lemma~\ref{lem:local_prime_power} (ii) implies $S/I \cong \GF{2^m}$
for some $m \in \{1,3\}$.
If $m = 3$, then $S \cong \GF{8}$,
and if $m = 1$, then by Corollary~\ref{cor:R_ideal},
$S \LSrdom \GF{2}$.
Similarly, a commutative local ring of size $4$ is either isomorphic to $\GF{4}$
or is dominated by $\GF{2}$.
Thus if $\Network$ is scalar linearly solvable over a ring $R$ 
satisfying case (c), then
$R \cong \GF{8} \times \GF{4}$;
otherwise, 
by Lemma~\ref{lem:direct_product},
a scalar linear solution over $R$ would imply there exists a scalar linear solution over $\GF{2}$.
Thus $\GF{8} \times \GF{4}$ is the only commutative ring of size $32$ 
over which $\Network$ is scalar linearly solvable.
\end{proof}

Theorem~\ref{thm:onlyRing32}
demonstrates that $\GF{8} \times \GF{4}$ is not dominated by any other ring of size $32$
and thus is maximal.
Theorem~\ref{thm:FieldBeatsRing}
demonstrates no finite field is dominated by any other ring of the same size,
and thus all finite fields are maximal.
In Section~\ref{sec:part_rings},
we characterize all maximal rings
and show that all maximal rings have the property
that there exists some network that is scalar linearly solvable over the maximal ring
but not over any other ring of the same size.

The network in the previous theorem is clearly also scalar linearly solvable over the fields $\GF{8}$ and $\GF{4}$.
So while $\GF{8} \times \GF{4}$ is the only commutative ring of size $32$ that the network is scalar linearly solvable over,
it is not the smallest commutative ring the network is scalar linearly solvable over.
This fact agrees with Theorem~\ref{thm:R_dom_by_smaller_field}.

\begin{theorem}
  Suppose a network is scalar linearly solvable over some commutative ring whose size is divisible by the prime $p$.
  Then the network is scalar linearly solvable over some finite field of characteristic $p$
  whose size divides the size of the ring.
  \label{thm:field_div_ring}
\end{theorem}
\begin{proof}
  Let the commutative ring be $R$.
  By Lemma~\ref{lem:dp_local_rings},
  there exist commutative local rings $R_1,\dots,R_n$ such that
  $$R \cong R_1 \DP \cdots \DP R_n.$$
  So we have
  $$ | R | = |R_1| \cdots |R_n| $$
  and since $p$ divides $|R|$,
  there exists $j \in \{1,\dots,n\}$
  such that $p$ divides $|R_j|$.
  By Lemma~\ref{lem:local_prime_power} (i),
  this implies $|R_j| = p^m$ for some positive integer $m$.
  Therefore, by Lemma~\ref{lem:GF_dom_local},
  $R_j \LSrdom \GF{p^m}$.

  Since $\Network$ is scalar linearly solvable over $R$,
  by Lemma~\ref{lem:direct_product},
  $\Network$ must be scalar linearly solvable over $R_j$,
  and since $R_j \LSrdom \GF{p^m}$,
  $\Network$ must also be scalar linearly solvable over $\GF{p^m}$.
\end{proof}

The following result shows that only commutative rings have square-free sizes.

\begin{lemma}{\cite[p. 512]{Eldridge-Orders}}
  Let $p_1,\dots,p_n$ be distinct primes.
  Every ring of size $p_1 \cdots p_n$ is
  commutative.
\label{lem:square-free-ring}
\end{lemma}

The following corollary shows that
if a network is scalar linearly solvable over a ring whose size is square-free,
then it must also be scalar linearly solvable over the prime fields
corresponding to its prime factors.

\begin{corollary}
Let $p_1,\dots,p_n$ be distinct primes.
If a network is scalar linearly solvable over a ring of size $p_1 \cdots p_n$,
then the network is scalar linearly solvable over each of the fields $\GF{p_1},\dots,$ $\GF{p_n}$.
\label{cor:ring-size-6}
\end{corollary}    
\begin{proof}
By Lemma~\ref{lem:square-free-ring},
every ring of size $p_1 \cdots p_n$ is commutative.
For each $i = 1,\dots,n$,
the only field of characteristic $p_i$ whose size divides $p_1\cdots p_n$ is $\GF{p_i}$,
so by Theorem~\ref{thm:field_div_ring},
any network that is scalar linearly solvable over a ring of size $p_1\cdots p_n$
is also scalar linearly solvable over $\GF{p_i}$.
\end{proof}

In general, one cannot specify in Theorem~\ref{thm:field_div_ring} which fields of characteristic $p$
a particular network is scalar linearly solvable over
without knowing the particular ring $R$.
As an example, the following corollary illustrates that
different networks that are scalar linearly solvable over different rings of size $12$,
may be scalar linearly solvable over different finite fields.

\begin{corollary}
  (i) If a network is scalar linearly solvable over $\GF{4} \times \GF{3}$,
  then the network is scalar linearly solvable over $\GF{4}$ and $\GF{3}$
  but not necessarily over $\GF{2}$.

  (ii) If a network is scalar linearly solvable over $\Z_{12}$,
  then the network is scalar linearly solvable over $\GF{2}$ and $\GF{3}$.
\label{cor:rings-size-12}
\end{corollary}
\begin{proof}
Part (i) follows from Lemma~\ref{lem:direct_product}
and the fact that the Two-Six Network
is scalar linearly solvable over $\GF{4}$ and $\GF{3}$ but not over $\GF{2}$
(see Corollary~\ref{cor:fourchoosetwo_rings}).

Part (ii) follows from Corollary~\ref{cor:ints_mod_m}.
\end{proof}

\clearpage
\section{Partition division}\label{sec:partition_div}

This section focuses on using integer partitions 
to describe a particular class of commutative rings
that are direct products of finite fields.
These rings will then be used in Section~\ref{sec:part_rings} to
characterize commutative rings that are maximal.

For any positive integer $k$,
a \textit{partition of} $k$
is a non-decreasing sequence of positive integers
$(a_1,\dots,a_r)$
whose sum is equal to $k$.
The \textit{length} of a partition $\pA$ is the number of elements in the sequence
and is denoted $|\pA|$.
Let $\Pi(k)$ denote the set of all partitions of $k$.
  
\begin{definition}
  For each prime $p$, 
  and each partition $\pA = (a_1,\dots,a_r)$ of $k$,
  define the product ring
  $$\PartitionRing{\pA}{p} =  \bigDP_{i = 1}^r \GF{p^{a_i}}.$$
  Let $m \ge 2$
  have prime factorization 
  $m = \PrimeFact{t}$,
  and let $R \in \mRings{m}$.
  We call $R$ a \textit{partition ring}
  if for each $i = 1,\dots,t$,
  there exists $\pA_i \in \Pi(k_i)$
  such that
  $$R \cong \bigDP_{i=1}^t \PartitionRing{\pA_i}{p_i}.$$
  We will refer to $\pA_1, \dots, \pA_t$ as the \textit{partitions of} $R$.
  \label{def:partition_field}
\end{definition}

As an example, if 
$m = 864 = 2^5 3^3$, 
then
$R = \GF{2^2} \times \GF{2^2} \times \GF{2^1} \times \GF{3^2} \times \GF{3^1}$
is a partition ring
and the partitions of $R$ are
$A_1 = (2, 2, 1)$ and
$A_2 = (2, 1)$.
Another partition ring of size $864$ is
$R = \GF{2^4} \times \GF{2^1} \times \GF{3^3}$
and the partitions of $R$ are
$A_1 = (4, 1)$ and
$A_2 = (3)$.

In later proofs, we will encounter direct products of fields
not given in terms of partitions;
however, Lemma~\ref{lem:dp_fields_iso_partition}
demonstrates that each such direct product is, in fact, a partition ring.

\begin{lemma}
  Every direct product of finite fields is a partition ring.
  \label{lem:dp_fields_iso_partition}
\end{lemma}
\begin{proof}
  Suppose $q_1,\dots,q_s$ are (not necessarily distinct) prime numbers
  and $n_1,\dots,n_s$ are positive integers and define the product ring
  \begin{align*}
    R &= \bigDP_{j=1}^s \GF{ q_j^{n_j} }  .
  \end{align*}
  Let $\PrimeFact{t}$ denote the prime factorization of the ring size $|R|$, so that
  \begin{align*}
    \PrimeFact{t} & = q_1^{n_1} \cdots q_s^{n_s}.
  \end{align*}
  For each $j \in \{1,\dots,s\}$,
  we have $q_j = p_i$ for some unique $i \in \{1,\dots,t\}$.
  Thus, for each $i = 1,\dots,t$, 
  there exist positive integers 
  $r_i$ and $\aa{i}{1} \ge \dots \ge \aa{i}{r_i}$
  such that
  $\displaystyle\sum_{j=1}^{r_i} \aa{i}{j} = k_i$ and
  \begin{align*}
    \bigDP_{j=1}^s \GF{ q_j^{n_j} }
     & \cong \bigDP_{i=1}^t \bigDP_{j=1}^{r_i} \GF{ p_i^{\aa{i}{j} }}.
  \end{align*}
  Let $\pA_i = (\aa{i}{1}, \dots, \aa{i}{r_i})$.
  Then for each $i$, $\pA_i$ is a partition of $k_i$, 
  and we have
  \begin{align*}
    R \cong \bigDP_{i=1}^{t} \bigDP_{j = 1}^{r_i} \GF{p_i^{\aa{i}{j}}} 
    &\cong \bigDP_{i=1}^t \PartitionRing{\pA_i}{p_i}.
  \end{align*}
\end{proof}

\begin{definition}
  Let $\pA$ and $\pB$ be partitions of $k$.
  We say that $\pB$ \textit{divides} $\pA$ and write $\pDiv{\pB}{\pA}$ if
  for each $a \in \pA$,
  there exists $b \in \pB$ 
  such that $\Div{b}{a}$.
  We call the relation ``$\pDivSymbol$'' \textit{partition division}.
  \label{def:partition_divide}
\end{definition}

For each positive integer $k$,
it can be verified that the partition division relation is a quasi-order on the set $\Pi(k)$.
Throughout this paper, whenever we refer to a partition of an integer as being maximal, we mean the partition is 
maximal with respect to the relation $\pDivSymbol$ on the set of all partitions of the same integer.
A partition $\pA$ of $k$ is maximal if and only if 
$\pDiv{\pB}{\pA}$ whenever $\pDiv{\pA}{\pB}$,
for all partitions $\pB$ of $k$.

Sometimes distinct partitions of the same integer each divide the other.
For example, for each $k \ge 3$, 
the partitions $(k-1,1)$ and $(k-2,1,1)$ of $k$ divide one another.
Hence partition division is not anti-symmetric on $\Pi(k)$.

The following lemma shows that if a partition divides a partition that is
not shorter than it, 
then it also divides a partition which is shorter.
This property will be used to characterize maximal partitions in 
Theorems~\ref{thm:part_stronger_maximal}
and \ref{thm:len_2_partition}.

\begin{lemma}
  Let $\pA$ and $\pB$ be different partitions of $k$.
  If $| \pA | \le | \pB |$ and $\pDiv{\pA}{\pB}$,
  then there exists a partition $\pC$ of $k$
  such that $| \pC | < | \pA |$
  and $\pDiv{\pA}{\pC}$.
  \label{lem:smaller_partitions}
\end{lemma}
\begin{proof}
  The proof uses induction on $|\pB| - |\pA|$.
  In this proof, when we refer to elements of an integer partition as being ``distinct''
  we mean that the elements are in different positions in the partition but possibly equal in value.

\begin{itemize}
  \item{Base case:} $|\pB| - |\pA| = 0$.\\
  It cannot be the case that each $b \in \pB$ has a distinct divisor $b' \in \pA$,
  for otherwise, $b' \le b$ for all $b\in \pB$, and we would have
  $$k = \sum_{b\in \pB} b \ge \sum_{b\in \pB} b' = k$$
  which would imply $b = b'$ for all $b\in \pB$,
  and thus $\pA = \pB$.
  So we may assume there exists $a \in \pA$ that divides
  some distinct elements $b_1,b_2 \in \pB$.
  Let $\pC$ be the partition $\pB$ with elements $b_1$ and $b_2$ removed
  and replaced by $(b_1 + b_2)$.
  Then $\pC$ is a partition of $k$ that is shorter than $\pA$,
  and since $a$ divides $(b_1 + b_2)$, 
  we have $\pDiv{\pA}{\pC}$.

  \item{Induction step:} Assume true whenever $|\pB| - |\pA| < n$ (where $n \ge 1$).\\
  Suppose $|\pB| - |\pA| = n$.
  \begin{itemize}
  \item[$\blacktriangleright$]{Case: $n=1$}\\
  Since $|\pB| > |\pA|$, there exists $a \in \pA$ that divides distinct $b_1,b_2 \in \pB$.
  If there is a third distinct element $b_3\in \pB$ such that $\Div{a}{b_3}$,
  then let $\pC$ be the partition $\pB$ with elements $b_1$, $b_2$, and $b_3$ removed
  and replaced by $(b_1 + b_2 + b_3)$.
  Then $\pC$ is a partition of $k$ that is shorter than $\pA$,
  and since $a$ divides $(b_1 + b_2 + b_3)$, 
  we have $\pDiv{\pA}{\pC}$.
  If there is no such third distinct element $b_3$, 
  then modify $\pB$ by removing the elements $b_1$ and $b_2$ 
  and adding an element $(b_1 + b_2)$.
  The new $\pB$ is a partition of $k$ that is the same length as $\pA$,
  and since $a$ divides $(b_1 + b_2)$, we have $\pDiv{\pA}{\pB}$.
  Since $a$ divides both $b_1$ and $b_2$, we have $a < b_1 + b_2$,
  and since $(b_1+b_2)$ is the only element of $\pB$ that $a$ divides, 
  the value $a$ is not one of the elements of $\pB$.
  Hence $\pB \ne \pA$,
  which reduces to the base case $n=0$.

  \item[$\blacktriangleright$]{Case: $n\ge 2$}\\
  Since $|\pB| > |\pA|$, there exists $a \in \pA$ that divides
  some distinct elements $b_1,b_2 \in \pB$.
  Modify the partition $\pB$ by removing the elements $b_1$ and $b_2$ 
  and adding the element $(b_1 + b_2)$.
  The new $\pB$ is a partition of $k$ that is one shorter than before the modification,
  and since $a$ divides $(b_1 + b_2)$, 
  we have $\pDiv{\pA}{\pB}$.
  This reduces to the case $|\pC| - |\pA| = n-1$, which is true by the induction hypothesis.
  \end{itemize}
\end{itemize}
\end{proof}

\begin{theorem}
  No maximal partition of $k$ can divide any other partition of $k$.
  \label{thm:part_stronger_maximal}
\end{theorem}
\begin{proof}
  Any partition $\pA$ is maximal if and only if the equivalence class $[\pA]$ is maximal 
  (with respect to the induced partial order under partition division),
  so it suffices to show that if $[\pA]$ is maximal,
  then $[\pA] = \{\pA\}$.

  Let $\pA$ be a maximal partition of $k$ 
  such that $\pA$ is of minimal length among the partitions in $[\pA]$,
  and suppose $\pB \in [\pA] - \{\pA\}$. 
  Then $| \pA | \le | \pB |$ and $\pDiv{\pA}{\pB}$,
  so by Lemma~\ref{lem:smaller_partitions},
  there exists $\pC \in \Pi(k)$
  such that $| \pC | < | \pA |$
  and $\pDiv{\pA}{\pC}$.
  Since $[\pA]$ is maximal, we must have 
  $\pDiv{\pC}{\pA}$, which implies $\pC \in [\pA]$,
  but this violates the minimum length of $\pA$ in $[\pA]$.
  Thus, $[\pA] = \{\pA\}$. 
\end{proof}

Theorem~\ref{thm:part_stronger_maximal} implies the maximal partitions of $k$
are precisely the partitions of $k$ 
that do not divide any other partition of $k$.
This is a stronger maximality condition than the maximality induced by the quasi-order.

Lemma~\ref{lem:antiprime_elements} demonstrates a property of maximal partitions
that will be used in a later proof.

\begin{lemma}
  No element of a maximal partition of $k$ is divisible by a different element of the partition.
  \label{lem:antiprime_elements}
\end{lemma}
\begin{proof}
  Let $\pA = (a_1,\dots,a_r)$ be a partition of $k$.
  Assume there exist distinct $i,j \in \{1,\dots,r\}$
  such that $a_i$ divides $a_j$.
  Then $a_i$ divides $(a_i + a_j)$.
  Create a new partition $\pB$ of $k$ by
  removing the elements $a_i$ and $a_j$ of $\pA$ 
  and inserting a new element $(a_i + a_j)$.
  Then $\pB \ne \pA$ and $\pDiv{\pA}{\pB}$,
  so by Theorem~\ref{thm:part_stronger_maximal}, $\pA$ is not maximal.
\end{proof}
The converse of Lemma~\ref{lem:antiprime_elements} does not necessarily hold.
For example, 
the partition $(5,3,2)$ satisfies the latter condition of Lemma~\ref{lem:antiprime_elements},
but $\pDiv{(5,3,2)}{(10)}$,
so $(5,3,2)$ is not maximal.

\subsection{Characterizing maximal partitions}\label{ssec:char_anti}

The following results provide a partial characterization of the maximal partitions
with respect to partition division.
%

\begin{remark}
  For each $k \ge 1$,
  the partition $(k)$ is maximal
  since $k$ does not divide any positive integer less than $k$.
  \label{rem:len_1_partition}
\end{remark}

Theorem~\ref{thm:len_2_partition} 
gives a complete characterization of the maximal partitions of length $2$.

\begin{theorem}
  Let $k$ and $m$ be positive integers such that $m \le k/2$.
  The partition $(k-m,m)$ of $k$ is maximal
  if and only if 
  $\NDiv{m}{k}$.
  \label{thm:len_2_partition}
\end{theorem}
\begin{proof}
  Assume $\Div{m}{k}$.
  Then $\pDiv{(k-m,m)}{(k)}$,
  so by Theorem~\ref{thm:part_stronger_maximal}, 
  $(k-m,m)$ is not a maximal partition.

  Now assume $\NDiv{m}{k}$.
  Then $k \ne 2m$, so $m < k/2$, or equivalently $k-m > k/2$.
  Thus, $\NDiv{(k-m)}{k}$,
  which means that $(k-m,m)$ does not divide $(k)$.
  But $(k)$ is the only partition of $k$ shorter than the partition $(k-m,m)$,
  so by Lemma~\ref{lem:smaller_partitions},
  the partition $(k-m,m)$ cannot divide any other partition of $k$ 
  that is at least as long as $(k-m,m)$.
  Thus $(k-m,m)$ is maximal.
\end{proof}

We can have maximal partitions of length $3$ or greater, such as $(7,6,4)$,
although we do not know of a nice characterization of such partitions.
In Table~\ref{tab:1} of the Appendix,
we provide a computer generated list of all maximal partitions of $k$,
for each $k \le 30$.

\begin{theorem}
  Let $k$ be a positive integer.
  Then $(k)$ is the unique maximal partition of $k$ 
  if and only if $k \in \{1,2,3,4,6\}$.
  \label{thm:tot_ordered}
\end{theorem}
\begin{proof}
  For each positive integer $k$,
  by Remark~\ref{rem:len_1_partition},
  $(k)$ is a maximal partition.
  It is easily verified that the following are all the partitions of
  $k$, for $k \in \{1,2,3,4,6\}$:
  \begin{align*}
  \Pi(1) &= \{(1)\}\\
  \Pi(2) &= \{(2), (1,1)\}\\
  \Pi(3) &= \{(3), (2,1), (1,1,1)\}\\
  \Pi(4) &= \{(4), (3,1), (2,2), (2,1,1), (1,1,1,1)\}\\
  \Pi(6) &= \left\{(6), (5,1), (4,2), (4,1,1), (3,3), (3,2,1), (3,1,1,1),\right.\\
         &\ \ \ \ \ \ \left.(2,2,2), (2,2,1,1), (2,1,1,1,1), (1,1,1,1,1,1)\right\}.
  \end{align*}
  For each $k \in \{1,2,3,4,6\}$,
  every partition of $k$ has an element that divides $k$,
  so $(k)$ is the only maximal partition for such $k$.

  For each odd $k \ge 5$,
  we have
  \begin{align*}
    \GCD{\frac{k-1}{2}}{k} &= \GCD{\frac{k-1}{2}}{k - 2 \, \left(\frac{k-1}{2} \right)} 
    = \GCD{\frac{k-1}{2}}{1} = 1 < \frac{k-1}{2}
  \end{align*}
  so $\NDiv{\frac{k-1}{2}}{k}$.
  Therefore 
  $\left(\frac{k+1}{2}, \frac{k-1}{2} \right)$ is a maximal partition,
  by taking $m=\frac{k-1}{2}$ in Theorem~\ref{thm:len_2_partition}.
  
  For each even $k \ge 8$,
  we have
  \begin{align*}
    \GCD{\frac{k}{2}-1}{k} &= \GCD{\frac{k}{2}-1}{k-2 \left(\frac{k}{2}-1 \right)} 
      = \GCD{\frac{k}{2}-1}{2} \le 2 < \frac{k}{2}-1
  \end{align*}
  so $\NDiv{\left(\frac{k}{2}-1 \right)}{k}$.
  Therefore 
  $\left(\frac{k}{2} + 1, \frac{k}{2} - 1 \right)$ is a maximal partition,
  by taking $m= \frac{k}{2}-1$ in Theorem~\ref{thm:len_2_partition}.
  
  Thus if $k = 5$ or if $k \ge 7$,
  then there exists at least two maximal partitions of $k$.
\end{proof}

\section{Characterizing maximal commutative rings}\label{sec:part_rings}

Lemma~\ref{lem:GF_partitions} demonstrates the connection between partition division
and dominance of partition rings.
Lemma~\ref{lem:GF_partitions} is a special case of Lemma~\ref{lem:direct_prod3},
where the direct products of finite fields are based on partition rings.

\begin{lemma}
  Let $m \ge 2$
  have prime factorization
  $m = \PrimeFact{t}$,
  and for each $i = 1,\dots,t$,
  let $\pA_i$ and $\pB_i$ be partitions of $k_i$.
  Then the ring
  $\PartitionRing{\pB_1}{p_1} \times \dots \times \PartitionRing{\pB_t}{p_t}$
  is dominated by the ring 
  $\PartitionRing{\pA_1}{p_1} \times \dots \times \PartitionRing{\pA_t}{p_t}$
  if and only if
  $\pB_i$ divides $\pA_i$ for all $i$.
  \label{lem:GF_partitions}
\end{lemma}
\begin{proof}
  For each $i \in \{1,\dots,t\}$,
  let $\pA_i = (\aa{i}{1},\dots,\aa{i}{r_i})$
  and $\pB_i = (\bb{i}{1},\dots,\bb{i}{s_i})$.
  Then
  $$\bigDP_{i=1}^t \PartitionRing{\pA_i}{p_i} 
  \cong \bigDP_{i=1}^t \bigDP_{j=1}^{r_i} \GF{p_i^{\aa{i}{j}} }
  \, \text{ and } \,
  \bigDP_{i=1}^t \PartitionRing{\pB_i}{p_i} 
  \cong \bigDP_{i=1}^t \bigDP_{j=1}^{s_i} \GF{p_i^{\bb{i}{j}} }. $$
  By Lemma~\ref{lem:direct_prod3},
  $$\bigDP_{i=1}^t \bigDP_{j=1}^{s_i} \GF{p_i^{\bb{i}{j}} } 
  \LSrdom
  \bigDP_{i=1}^t \bigDP_{j=1}^{r_i} \GF{p_i^{\aa{i}{j}} }$$
  if and only if
  for each $i \in \{ 1,\dots,t\}$
  and each $j \in \{1,\dots,r_i\}$,
  there exists $l \in \{1,\dots,s_i\}$
  such that $\Div{ \bb{i}{l} }{ \aa{i}{j} }$.
  However, the latter condition 
  is precisely
  $\pDiv{\pB_i}{\pA_i}$ for all $i$.
\end{proof}

\begin{corollary}
  If each of a partition ring's integer partitions is maximal, 
  then the ring is not dominated by any other partition ring of the same size.
  \label{cor:max_dom_part_rings}
\end{corollary}
\begin{proof}
  Let $m = \PrimeFact{t}$ be the prime factorization of $m$.
  For each $i = 1,\dots,t$,
  let $\pA_i,\pB_i \in \Pi(k_i)$
  be such that $\pA_i$ is maximal.
  Suppose
  $$ \bigDP_{i=1}^t \PartitionRing{\pA_i}{p_i} \LSrdom \bigDP_{i=1}^t \PartitionRing{\pB_i}{p_i} .$$
  Then by Lemma~\ref{lem:GF_partitions},
  $\pDiv{\pA_i}{\pB_i}$
  for all $i$.
  Since each $\pA_i$ is maximal,
  by Theorem~\ref{thm:part_stronger_maximal},
  $\pB_i = \pA_i$,
  for all $i$.
  Therefore
  $$ \bigDP_{i=1}^t \PartitionRing{\pB_i}{p_i} \cong \bigDP_{i=1}^t \PartitionRing{\pA_i}{p_i} .$$
\end{proof}

Lemma~\ref{lem:anti_not_dominated} extends Corollary~\ref{cor:max_dom_part_rings}
to show that partition rings, where each partition is maximal,
are not dominated by any other (not necessarily partition) commutative ring of the same size.

\begin{lemma}
  If each of a partition ring's integer partitions is maximal, 
  then the ring is not dominated by any other commutative ring of the same size.
  \label{lem:anti_not_dominated}
\end{lemma}
\begin{proof}
  Let $m = \PrimeFact{t}$ be the prime factorization of the size of the ring 
  $$R = \bigDP_{i=1}^t \PartitionRing{\pA_i}{p_i}$$
  where for each $i = 1,\dots,t$,
  the partition $\pA_i = (\aa{i}{1},\dots,\aa{i}{r_i})$ of $k_i$ 
  is maximal.
  Suppose $R$ is dominated by a commutative ring $S$ of size $m$. 
  We will show that $R$ and $S$ are isomorphic rings.

  By Lemma~\ref{lem:dp_local_rings},
  $S$ can be written as a direct product of commutative local rings,
  and by Lemma~\ref{lem:local_prime_power} (i), 
  the size of each such local ring has to be a power of one of the prime factors
  $p_1, \dots, p_t$ of $m$.
  Specifically,
  for each $i = 1,\dots,t$,
  there exist local rings $L_{i,1},\dots,L_{i,s_i}$ 
  such that each $|L_{i,j}|$ is a power of $p_i$
  and
  \begin{align}
    S & \cong \bigDP_{i=1}^t \bigDP_{j=1}^{s_i} L_{i,j}
    \label{eq:anti_0} .
  \end{align}
  For each $i = 1,\dots,t$ and $j = 1,\dots,s_i$,
  Lemma~\ref{lem:GF_dom_local} impies that
  $L_{i,j} \LSrdom \GF{ |L_{i,j}| } $.
  Then,
  \begin{align}
    \bigDP_{i=1}^t \PartitionRing{\pA_i}{p_i}
    &\LSrdom 
    \bigDP_{i=1}^t \bigDP_{j=1}^{s_i} L_{i,j}
      & & \Comment{$R \LSrdom S$, \eqref{eq:anti_0}}\notag\\
    & \LSrdom 
    \bigDP_{i=1}^t \bigDP_{j=1}^{s_i} \GF{ |L_{i,j}| } 
      & & \Comment{Lemma~\ref{lem:direct_prod2}} 
    \label{eq:anti_2}
  \end{align}
  and the right-hand-side of \eqref{eq:anti_2} is a partition ring of size $m$,
  by Lemma~\ref{lem:dp_fields_iso_partition}.

  Since each $\pA_i$ is maximal,
  by Corollary~\ref{cor:max_dom_part_rings} and \eqref{eq:anti_2},
  we have
  \begin{align}
      \bigDP_{i=1}^t \bigDP_{j=1}^{s_i} \GF{ |L_{i,j}| } 
        & \cong \bigDP_{i=1}^t \PartitionRing{\pA_i}{p_i}  \notag \\
        & \cong \bigDP_{i=1}^t \bigDP_{j=1}^{r_i} \GF{p_i^{ \aa{i}{r_j} } }
      \label{eq:anti_4} .
  \end{align}
  Therefore for each $i = 1,\dots,t$,
  we have $s_i = r_i$,
  and by \eqref{eq:anti_4},
  without loss of generality,
  we may assume $|L_{i,j}| = p_i^{\aa{i}{j}}$,
  for all $j = 1,\dots,r_i$.

  For each $i = 1,\dots,t$ and $j = 1,\dots,r_i$,
  let $I_{i,j}$ be the maximal ideal of the local ring $L_{i,j}$.
  Then, by Lemma~\ref{lem:local_prime_power} (ii),
  for each $i$ and $j$,
  there exists a positive integer $\bb{i}{j}$
  such that 
  $\Div{\bb{i}{j}}{\aa{i}{j}}$
  and
  $\GF{ p_i^{\bb{i}{j}} } \cong L_{i,j} / I_{i,j}$.
  Corollary~\ref{cor:R_ideal} then implies
  \begin{align}
    L_{i,j} \LSrdom \GF{ p_i^{\bb{i}{j}} }
      & &  (i = 1,\dots,t \text{ and } j = 1,\dots,r_i)
      \label{eq:anti_6} 
  \end{align}
  and therefore
  \begin{align}
      \bigDP_{i = 1}^{t} \bigDP_{j = 1}^{r_i} \GF{ p_i^{\aa{i}{j}} } 
      & \LSrdom 
      \bigDP_{i=1}^t \bigDP_{j=1}^{r_i} L_{i,j} 
        & & \Comment{$R \LSrdom S$, \eqref{eq:anti_0}}\notag\\
      & \LSrdom 
      \bigDP_{i=1}^t \bigDP_{j=1}^{r_i} \GF{ p_i^{\bb{i}{j}} } 
        & & \Comment{\eqref{eq:anti_6}, Lemma~\ref{lem:direct_prod2}} 
      \label{eq:anti_3}.
  \end{align}

  Lemma~\ref{lem:direct_prod3} and
  \eqref{eq:anti_3} imply that 
  for each $i \in \{1,\dots,t\}$ and $j \in \{1,\dots,r_i\}$,
  there exists $l \in \{1,\dots,r_i\}$
  such that 
  $\Div{\aa{i}{l}}{\bb{i}{j}}$.
  We also have $\Div{\bb{i}{j}}{\aa{i}{j}}$, so
  $\Div{\aa{i}{l}}{\aa{i}{j}}$.
  Since $\pA_i$ is maximal,
  by Lemma~\ref{lem:antiprime_elements},
  this implies $l = j$.
  Thus $\bb{i}{j} = \aa{i}{j}$, 
  for all $i \in \{1,\dots,t\}$ and $j \in \{1,\dots,r_i\}$,
  and therefore
  $L_{i,j} / I_{i,j} \cong \GF{ p_i^{\aa{i}{j}} }$
  for all $i,j$.
  However,
  we also have $|L_{i,j}| = p_i^{\aa{i}{j}}$
  for all $i,j$. 
  So it must be the case that
  $|I_{i,j}| = 1$,
  and 
  \begin{align}
    L_{i,j} \cong \GF{ p_i^{\aa{i}{j}}} & & (i = 1,\dots,t \text{ and } j = 1,\dots,r_i) .
      \label{eq:anti_5}
  \end{align}
  Thus,
  \begin{align*}
    S &\cong \bigDP_{i=1}^t \bigDP_{j=1}^{r_i} \GF{p_i^{\aa{i}{j}}} 
      & & \Comment{\eqref{eq:anti_0}, \eqref{eq:anti_5}} \notag\\
      &\cong \bigDP_{i=1}^t \PartitionRing{\pA_i}{p_i} \cong R.
  \end{align*}
\end{proof}

Lemmas~\ref{lem:anti_not_dominated} and \ref{lem:DP_GF_dom_ring} 
will be used in the proof of Theorem~\ref{thm:maximal_rings} 
to show that the maximal commutative rings with respect to dominance
are precisely partition rings where each partition is maximal.

\begin{lemma}
Every finite commutative ring is dominated by some partition ring of the same size,
all of whose partitions are maximal.
  \label{lem:DP_GF_dom_ring}
\end{lemma}
\begin{proof}
  Let $R$ be a finite commutative ring.
  By Lemma~\ref{lem:dp_local_rings},
  there exist
  commutative local rings $R_1,R_2,\dots,R_n$ such that
  \begin{align}
    R \cong \bigDP_{j=1}^n R_j
    \label{eq:0}.
  \end{align}
  By Lemma~\ref{lem:GF_dom_local},
  for each $j = 1,\dots,n$,
  we have
  $R_j \LSrdom \GF{|R_j|}$,
  so by Lemma~\ref{lem:direct_prod2}, we have
  \begin{align}
    \bigDP_{j=1}^n R_j \LSrdom \bigDP_{j=1}^n \GF{|R_j|}.
    \label{eq:1}
  \end{align}

  Let $m = \PrimeFact{t}$ denote the prime factorization of $m$.
  Then by Lemma~\ref{lem:dp_fields_iso_partition},
  for each $i = 1,\dots,t$,
  there exists a partition $\pB_i$ of $k_i$
  such that
  \begin{align}
    \bigDP_{i=1}^t \PartitionRing{\pB_i}{p_i} & \cong \bigDP_{j=1}^n \GF{|R_j|} .
      \label{eq:2}
  \end{align}
  
  Since $\Pi(k_i)$ is a finite quasi-ordered set under \pdiv,
  for each $i = 1,\dots,t$,
  there exists maximal $\pA_i \in \Pi(k_i)$
  such that $\pDiv{\pB_i}{\pA_i}$.
  So we have
  \begin{align*}
  R & \cong \bigDP_{j=1}^n R_j
      & & \Comment{\eqref{eq:0}}\\
    & \LSrdom \bigDP_{j=1}^n \GF{|R_j|}
      & & \Comment{\eqref{eq:1}} \\
    & \cong \bigDP_{i=1}^t \PartitionRing{\pB_i}{p_i}
      & & \Comment{\eqref{eq:2}} \\
    & \LSrdom \bigDP_{i=1}^t \PartitionRing{\pA_i}{p_i} 
      & & \Comment{Lemma~\ref{lem:GF_partitions}}.
  \end{align*}
\end{proof}

The following theorem characterizes maximal commutative rings.

\begin{theorem}
  A finite commutative ring is maximal
  if and only if
  it is a partition ring, each of whose integer partitions is maximal.
  \label{thm:maximal_rings}
\end{theorem}
\begin{proof}
  If $R$ is a partition ring such that each of its partitions is maximal,
  then by Lemma~\ref{lem:anti_not_dominated},
  no other commutative ring of the same size dominates $R$.
  Thus, $R$ is maximal.

  Conversely, assume commutative ring $R$ is maximal.
  By Lemma~\ref{lem:DP_GF_dom_ring}, 
  $R$ is dominated by a partition ring $S$ of the same size where each of its partitions is maximal.
  Since $R$ is maximal, this implies $S \LSrdom R$. 
  However, by Lemma~\ref{lem:anti_not_dominated},
  this implies $S \cong R$.
  Thus, $R$ is a partition ring such that each of its partitions is maximal.
\end{proof}

\begin{corollary}
  Let $m \ge 2$
  have prime factorization
  $m = \PrimeFact{t}$.
  Then $\GF{p_1^{k_1}} \times \dots \times \GF{p_t^{k_t}}$
  is a maximal ring of size $m$.
\label{cor:prod_max}
\end{corollary}
\begin{proof}
This follows from Theorem~\ref{thm:maximal_rings}
and Remark~\ref{rem:len_1_partition}.
\end{proof}

\begin{corollary}
  No maximal commutative ring is dominated by any other commutative ring of the same size.
  \label{cor:maximal_not_dominated}
\end{corollary}
\begin{proof}
  This follows immediately from Theorem~\ref{thm:maximal_rings}
  and Lemma~\ref{lem:anti_not_dominated}
\end{proof}
We note that this is a stronger maximality than
the maximality induced by the quasi-order.

Theorem~\ref{thm:FieldBeatsRing} demonstrated that for each finite field,
there exists a multicast network that is scalar linearly solvable over the field
but not over any other commutative ring of the same size,
and Theorem~\ref{thm:onlyRing32} demonstrated a network that is scalar linearly solvable over
$\GF{8} \times \GF{4}$ but not over any other commutative ring of size $32$.
The following theorem shows a similar property for every maximal commutative ring
and provides an alternate characterization of maximal commutative rings than
in Theorem~\ref{thm:maximal_rings}.

\begin{theorem}
  A finite commutative ring is maximal 
  if and only if
  there exists a network that is scalar linearly solvable over the ring
  but not over any other commutative ring of the same size.
  \label{thm:best_ring_for_network}
\end{theorem}
\begin{proof}
  Let $R$ be a maximal commutative ring of size $m$.
  By Corollary~\ref{cor:maximal_not_dominated},
  $R$ is not dominated by any other commutative ring of size $m$,
  so for each ring $S$ of size $m$ that is not isomorphic to $R$, 
  there exists a network $\Network_S$
  that is scalar linearly solvable over $R$ but not $S$.
  Then the disjoint union of networks
  $$\bigcup_{\substack{S \in \mRings{m} \\ S \not\cong R}} \Network_S$$
  is scalar linearly solvable over $R$,
  since each $\Network_S$ is scalar linearly solvable over $R$.
  However, for each $S \in \mRings{m}$, if $S$ is not isomorphic to $R$, then
  $\Network_S$ is not scalar linearly solvable over $S$,
  so the disjoint union of networks
  $$\bigcup_{\substack{S \in \mRings{m} \\ S \not\cong R}} \Network_S $$
  is not scalar linearly solvable over $S$.

  Conversely, if $R$ is a finite commutative ring that is not maximal,
  then, it is dominated by some other commutative ring $S$ of the same size,
  so any network that is scalar linearly solvable over $R$
  is also scalar linearly solvable over $S$.
\end{proof}

An interesting open problem related to Theorem~\ref{thm:best_ring_for_network}
is to characterize rings with the property that
there exists a multicast network that is scalar linearly solvable over the ring
but not over any other commutative ring of the same size.
We showed (in Theorem~\ref{thm:FieldBeatsRing})
that such a multicast network exists for every finite field,
and we showed (in Example~\ref{ex:size213ring})
that there exists a multicast network that is scalar linearly solvable over a ring
of size $2^{13}$ but not the field $\GF{2^{13}}$.

Theorem~\ref{thm:p5_7} demonstrates that in some cases,
there is only one maximal commutative ring of a given size.
If $R$ is the only maximal ring of a given size,
then by Lemma~\ref{lem:DP_GF_dom_ring}, 
any network with a scalar linear solution over some ring of size $|R|$
also has a scalar linear solution over $R$.
Alternatively, 
since the set of commutative rings of size $|R|$
is finite and quasi-ordered under dominance,
each ring $S \in \mRings{|R|}$
is dominated by some maximal ring,
and if $R$ is the only maximal ring of size $|R|$,
then $S$ is dominated by $R$.
In this case, $R$ can be thought of as the ``best'' commutative ring
of size $|R|$, in terms of maximal scalar linear solvability.

However, when there are multiple maximal rings of a given size,
not every network with a scalar linear solution over some ring of this size
is scalar linearly solvable over every maximal ring,
since by Theorem~\ref{thm:best_ring_for_network},
for each maximal ring,
there exists a network which
is scalar linearly solvable over the maximal ring
but not over any other commutative ring of the same size.
Thus, when there is more than one maximal ring of a given size,
there is no ``best'' commutative ring of this size.

\begin{theorem}
  Let $m \ge 2$
  have prime factorization
  $m = \PrimeFact{t}$.
  Then $\GF{p_1^{k_t}} \times \dots \times \GF{p_t^{k_t}}$ is the only maximal ring of size $m$
  if and only if
  $\{k_1,\dots,k_t\} \subseteq \{1,2,3,4,6\}$.
  \label{thm:p5_7}
\end{theorem}
\begin{proof}
  By Corollary~\ref{cor:prod_max}, $\displaystyle\bigDP_{i=1}^{t} \GF{p_i^{k_i}}$ is a maximal ring.
  Assume $k_i \in \{1,2,3,4,6\}$ for all $i$.
  Then by Theorem~\ref{thm:tot_ordered},
  $(k_i)$ is the only maximal partition of $k_i$ for all $i$.
  Thus, by Theorem~\ref{thm:maximal_rings}, 
  $\displaystyle\bigDP_{i=1}^{t} \GF{p_i^{k_i}}$
  is the only maximal ring of size $m$.
  
  Conversely, assume there exists $j$
  such that $k_j = 5$ or $k_j \ge 7$.
  Then by Theorem~\ref{thm:tot_ordered},
  there exists a maximal partition $\pB_j$ of $k_j$
  such that $\pB_j \ne (m)$.
  Then by Theorem~\ref{thm:maximal_rings}, 
  $\displaystyle\PartitionRing{\pB_j}{p_j} \DP \bigDP_{\substack{i = 1 \\ i \ne j}}^{t} \GF{p_i^{k_i}}$
  and 
  $\displaystyle\bigDP_{i=1}^{t} \GF{p_i^{k_i}}$
  are distinct maximal rings of size $m$.
\end{proof}

The bound in the following corollary can be achieved with equality, as illustrated in Example~\ref{ex:size32ring}.

\begin{corollary}
  If a network is not scalar linearly solvable over a given finite field
  but is scalar linearly solvable over some commutative ring of the same size,
  then the size of the field is at least $32$.
  \label{cor:32_min_size}
\end{corollary}
\begin{proof}
  It follows from Theorem~\ref{thm:p5_7}
  that for each $k \in \{1,2,3,4,6\}$ and prime $p$,
  any network that is scalar linearly solvable over some commutative ring of size $p^k$
  must also be scalar linearly solvable over the field $\GF{p^k}$.
  The claim follows from the fact $p = 2$ and $k = 5$
  yield the minimum $p^k$ that does not satisfy this condition.
\end{proof}

In the following example, we list the maximal rings of various sizes.

\begin{example}
  For each integer $k \ge 1$ and prime $p$,
  $\GF{p^k}$ is a maximal ring.
  The following are the other maximal commutative rings of size $p^k$ for all $k \leq 12$:
  \begin{itemize}
    \item $p^5: \;$ $\GF{p^3} \times \GF{p^2}$
    \item $p^7: \;$ $\GF{p^5} \times \GF{p^2}$ 
      and $\GF{p^4} \times \GF{p^3}$
    \item $p^8: \;$ $\GF{p^5} \times \GF{p^3}$
    \item $p^9: \;$ $\GF{p^7} \times \GF{p^2}$ 
      and $\GF{p^5} \times \GF{p^4}$
    \item $p^{10}: \;$ $\GF{p^7} \times \GF{p^3}$ 
      and $\GF{p^6} \times \GF{p^4}$
    \item $p^{11}: \; $ $\GF{p^9} \times \GF{p^2}$, 
      $\GF{p^8} \times \GF{p^3}$, 
      $\GF{p^7} \times \GF{p^4}$, 
      and $\GF{p^6} \times \GF{p^5}$
    \item $p^{12}: \;$ $\GF{p^7} \times \GF{p^5}$.
  \end{itemize}
  $\GF{8} \times \GF{4}$ is the smallest prime-power size maximal commutative ring that is not a finite field,
  and $\GF{128} \times \GF{64} \times \GF{16}$ has size $2^{17}$
  and is the smallest known%
  \footnote{
  If there were a prime-power size maximal commutative ring, consisting of a direct product of more than two fields,
  and whose size were less than $2^{17}$,
  then there would exist a length-$3$ maximal partition of an integer less than $17$.
  The enumeration of maximal partitions given in Table~\ref{tab:1}
  implies such a partition does not exist.
  }
   prime-power size maximal commutative ring consisting of a direct product of more than two fields.

  Maximal commutative rings of non-power-of-prime size are direct products of maximal commutative rings of prime-power size
  and can be found using the maximal partitions of the prime factor multiplicities. 
  For example, consider maximal rings of size $777 600 =2^7 3^5 5^2$.
  The maximal partitions of $7$ are $(7),(5,2),$ and $(4,3)$;
  the maximal partitions of $5$ are $(5)$ and $(3,2)$;
  and the only maximal partition of $2$ is $(2)$.
  Hence the $6$ maximal commutative rings of size $777 600$ are
  \begin{align*}
   \PartitionRing{(7)}{2} \times \PartitionRing{(5)}{3} \times \PartitionRing{(2)}{5}  
      &= \GF{2^7} \times \GF{3^5} \times \GF{5^2}  \\
   \PartitionRing{(5,2)}{2} \times \PartitionRing{(5)}{3} \times \PartitionRing{(2)}{5}
      &= \GF{2^5} \times \GF{2^2} \times \GF{3^5} \times \GF{5^2} \\
    \PartitionRing{(4,3)}{2} \times \PartitionRing{(5)}{3} \times \PartitionRing{(2)}{5}
      &= \GF{2^4} \times \GF{2^3} \times \GF{3^5} \times \GF{5^2} \\
    \PartitionRing{(7)}{2} \times \PartitionRing{(3,2)}{3} \times \PartitionRing{(2)}{5}
      &= \GF{2^7} \times \GF{3^3} \times \GF{3^2} \times \GF{5^2} \\
    \PartitionRing{(5,2)}{2} \times \PartitionRing{(3,2)}{3} \times \PartitionRing{(2)}{5}
      &= \GF{2^5} \times \GF{2^2} \times \GF{3^3} \times \GF{3^2} \times \GF{5^2} \\
    \PartitionRing{(4,3)}{2} \times \PartitionRing{(3,2)}{3} \times \PartitionRing{(2)}{5}
      &= \GF{2^4} \times \GF{2^3} \times \GF{3^3} \times \GF{3^2} \times \GF{5^2} .
  \end{align*}
  Table~\ref{tab:1} provides a list of the maximal partitions of $k$
  for $k = 1,2,\dots,30$,
  which can be used to find maximal commutative rings of size $m = \PrimeFact{t}$,
  where $k_1,\dots,k_t \leq 30$.
\label{ex:maximalrings_p10}
\end{example}

\clearpage
\section{Appendix}\label{sec:A}

\begin{table}[h!]
\small
 \begin{center}
  \begin{tabular}{|l|}
  \hline
  (1) 
  \\ \hline
  (2) 
  \\ \hline
  (3) 
  \\ \hline
  (4) 
  \\ \hline
  (5) \
    (3,2) 
  \\ \hline
  (6) 
  \\ \hline
  (7) \ 
    (5,2) \
    (4,3) 
  \\ \hline
  (8) \
    (5,3) 
  \\ \hline
  (9) \ 
      (7,2) \
      (5,4) 
  \\ \hline
  (10) \
      (7,3) \
      (6,4) 
  \\ \hline
  (11) \
      (9,2) \
      (8,3) \ 
      (7,4) \
      (6,5) 
  \\ \hline
  (12) \
      (7,5) 
  \\ \hline
  (13) \ 
      (11,2) \ 
      (10,3) \
      (9,4) \
      (8,5) \ 
      (7,6) 
  \\ \hline
  (14) \
      (11,3) \
      (10,4) \ 
      (9,5) \
      (8,6) 
  \\ \hline
  (15) \
      (13,2) \
      (11,4) \
      (9,6) \
      (8,7) 
  \\ \hline
  (16) \
      (13,3) \
      (11,5) \
      (10,6) \
      (9,7) 
  \\ \hline
  (17) \
      (15,2) \
      (14,3) \
      (13,4) \
      (12,5) \
      (11,6) \
      (10,7) \
      (9,8) \
      (7,6,4) 
  \\ \hline
  (18) \
      (14,4) \
      (13,5) \
      (11,7) \
      (10,8) 
  \\ \hline
  (19) \
      (17,2) \
      (16,3) \
      (15,4) \
      (14,5) \
      (13,6) \
      (12,7) \
      (11,8) \
      (10,9) \
      (9,6,4) \
      (8,6,5) 
  \\ \hline
  (20) \
      (17,3) \
      (14,6) \
      (13,7) \
      (12,8) \
      (11,9) 
  \\ \hline
  (21) \
      (19,2) \
      (17,4) \
      (16,5) \
      (15,6) \
      (13,8) \
      (12,9) \
      (11,10) \
      (11,6,4) 
  \\ \hline
  (22) \
      (19,3) \
      (18,4) \
      (17,5) \
      (16,6) \
      (15,7) \
      (14,8) \
      (13,9) \
      (12,10) \
      (9,8,5) \
      (9,7,6) 
  \\ \hline
  (23) \
      (21,2) \
      (20,3) \
      (19,4) \
      (18,5) \
      (17,6) \
      (16,7) \
      (15,8) \
      (14,9) \
      (13,10) \ 
      (13,6,4) 
      \\ \hspace{7.5mm}
      (12,11) \
      (11,7,5) \ 
      (10,9,4) \ 
      (10,7,6) \
      (9,8,6) 
  \\ \hline
  (24) \
    (19,5) \
    (17,7) \
    (15,9) \
    (14,10) \
    (13,11) \ 
  \\ \hline
  (25) \ 
    (23,2) \ 
    (22,3) \
    (21,4) \
    (19,6) \
    (18,7) \
    (17,8) \
    (16,9) \ 
    (15,10) \ 
    (15,6,4) \ 
    (14,11) 
    \\ \hspace{7.5mm}
    (13,12) \
    (11,10,4) \ 
    (11,8,6) \
    (10,9,6) \
    (10,8,7) \
  \\ \hline
  (26) \
    (23,3) \ 
    (22,4) \
    (21,5) \
    (20,6) \
    (19,7) \
    (18,8) \
    (17,9) \
    (16,10) \ 
    (15,11) 
    \\ \hspace{7.5mm}
    (14,12) \
    (12,9,5) \ 
    (11,9,6) \
    (11,8,7) \
    (10,9,7) \
  \\ \hline
  (27) \
    (25,2) \
    (23,4) \
    (22,5) \
    (21,6) \
    (20,7) \
    (19,8) \
    (17,10) \
    (17,6,4) \ 
    (16,11) 
    \\ \hspace{7.5mm}
    (15,12) \
    (14,13) \
    (14,8,5) \ 
    (13,10,4) \ 
    (13,8,6) \
    (12,8,7) \
    (11,10,6) \ 
  \\ \hline
  (28) \
    (25,3) \ 
    (23,5) \
    (22,6) \
    (20,8) \
    (19,9) \
    (18,10) \ 
    (17,11) \
    (16,12) \
    (15,13) 
    \\ \hspace{7.5mm}
    (13,9,6) \ 
    (12,11,5) \ 
    (11,9,8) \
  \\ \hline
  (29) \
    (27,2) \
    (26,3) \
    (25,4) \
    (24,5) \
    (23,6) \
    (22,7) \
    (21,8) \
    (20,9) \
    (19,10) \
    (19,6,4) 
    \\ \hspace{7.5mm}
    (18,11) \
    (17,12) \
    (16,13) \
    (16,7,6) \ 
    (15,14) \
    (15,10,4) \ 
    (15,8,6) \
    (14,11,4) 
    \\ \hspace{7.5mm}
    (14,9,6) \
    (13,11,5) \
    (13,10,6) \ 
    (13,9,7) \
    (12,10,7) \ 
    (12,9,8) \
    (11,10,8) \
  \\ \hline
  (30) \
    (26,4) \
    (23,7) \
    (22,8) \
    (21,9) \
    (19,11) \
    (18,12) \
    (17,13) \
    (16,14) \
    (13,9,8) \
    (12,11,7) \
  \\ \hline
     \end{tabular}
 \end{center}
  \caption{The maximal partitions of $k = 1,2,\dots,30$ under partition division.}
  \label{tab:1}
\end{table}

\clearpage

\renewcommand{\baselinestretch}{1.0}

\typeout{}
\typeout{------------------------------------------------------------------}
\typeout{Use the values below for the NonCommutative paper cross references}

\typeout{\unexpanded{\PartOne }______________________ \PartOne }
\typeout{\unexpanded{\PartTwo }______________________ \PartTwo }

\typeout{\unexpanded{\PartOneSectionModel}\getrefnumber{ssec:model}                 \unexpanded{ Section {ssec:model} } }
\typeout{\unexpanded{\PartOneTheoremSmallerField}\getrefnumber{thm:R_dom_by_smaller_field} \unexpanded{ Theorem {thm:R_dom_by_smaller_field} } }
\typeout{\unexpanded{\PartOneLemmaDirectProduct}\getrefnumber{lem:direct_product}         \unexpanded{ Lemma  {lem:direct_product}  } }
\typeout{\unexpanded{\PartOneTheoremFieldBeatsRing}\getrefnumber{thm:FieldBeatsRing}         \unexpanded{Theorem {thm:FieldBeatsRing}  } }
\typeout{\unexpanded{\PartOneTheoremOnlyThirtyTwo}\getrefnumber{thm:onlyRing32}           \unexpanded{ Theorem   {thm:onlyRing32} } }
\typeout{\unexpanded{\PartOnePartDomRing}\getrefnumber{lem:DP_GF_dom_ring}         \unexpanded{ Lemma   {lem:DP_GF_dom_ring} } }
\typeout{\unexpanded{\PartOneTheoremBestRingNetwork}\getrefnumber{thm:best_ring_for_network}  \unexpanded{ Theorem {thm:best_ring_for_network} } }
\typeout{\unexpanded{\PartOneTheoremPFiveSeven}\getrefnumber{thm:p5_7}                   \unexpanded{ Theorem {thm:p5_7} } }
\typeout{\unexpanded{\PartOneNChooseTwoFigure}\getrefnumber{fig:ChooseTwo}                   \unexpanded{ Figure {fig:ChooseTwo} } }
\typeout{\unexpanded{\PartOneLemmaSubring}\getrefnumber{lem:subring}                   \unexpanded{ Lemma {lem:subring} } }
\typeout{\unexpanded{\PartOneCorollaryIdeal}\getrefnumber{cor:R_ideal}                   \unexpanded{ Corollary {cor:R_ideal} } }
\typeout{------------------------------------------------------------------}

\end{document}